\documentclass[a4paper,UKenglish,cleveref,autoref, thm-restate]{lipics-v2021}
\pdfoutput=1

\usepackage{arxiv-lipics-v2021} 

\usepackage{mathtools} 
\usepackage[dvipsnames]{xcolor}
\usepackage[colorinlistoftodos]{todonotes}
\usepackage{CJKutf8} 
\usepackage{tikz}
\usetikzlibrary{arrows.meta, positioning}
\usepackage{xspace}
\usepackage{graphicx}
\usetikzlibrary{arrows, automata}
\usepackage{hyperref}
\usepackage[capitalise]{cleveref}
\usepackage{enumerate}
\usepackage{thmtools}
\usepackage{thm-restate}
\usepackage{environ}
\usepackage{tabularx}
\usepackage{comment}
\usepackage[linesnumbered]{algorithm2e} 
\usepackage[noend]{algpseudocode}

\usepackage[colorinlistoftodos]{todonotes}

\title{Exploring VASS Parameterised by Geometric Dimension}

\author{Wojciech Czerwi\'nski}
{University of Warsaw \and \url{https://www.mimuw.edu.pl/~wczerwin}}
{wczerwin@mimuw.edu.pl}
{https://orcid.org/0000-0002-6169-868X}
{}

\author{Roland Guttenberg}
{University of Warsaw}
{r.guttenberg@uw.edu.pl}
{https://orcid.org/0000-0001-6140-6707}
{}

\author{Łukasz Orlikowski}
{University of Warsaw}
{l.orlikowski@mimuw.edu.pl}
{https://orcid.org/0009-0001-4727-2068}
{}

\author{Henry Sinclair-Banks}
{University of Warsaw \and \url{http://henry.sinclair-banks.com}}
{hsb@mimuw.edu.pl}
{https://orcid.org/0000-0003-1653-4069}
{}

\author{Yangluo Zheng}
{Shanghai Jiao Tong University}
{wunschunreif@sjtu.edu.cn}
{https://orcid.org/0009-0000-1028-5458}
{}

\authorrunning{W. Czerwi{\'n}ski, R. Guttenberg, {\L}. Orlikowski, H. Sinclair-Banks, and Y. Zheng}

\Copyright{Wojciech Czerwi{\'n}ski, Roland Guttenberg, {\L}ukasz Orlikowski, Henry Sinclair-Banks, and Yangluo Zheng}



\category{}

\relatedversion{}

\funding{First four authors are supported by the ERC grant INFSYS, agreement no. 950398. Yangluo Zheng is supported by the NSFC (Natural Science Foundation of China) grant no. 62572319.}

\nolinenumbers 

\newcommand{\Oo}{\mathcal{O}}
\newcommand{\oo}{\Oo}

\newcommand{\NORM}{\textup{\textsc{norm}}}
\newcommand{\poly}{\textup{poly}}

\newcommand{\rank}{\textup{rank}}
\newcommand{\size}{\textup{size}}

\newcommand{\abs}[1]{\lvert#1\rvert}
\newcommand{\set}[1]{\{#1\}}
\newcommand{\sset}{\subseteq}
\newcommand{\N}[0]{\mathbb{N}}
\newcommand{\NN}{\mathbb{N}}
\newcommand{\Q}[0]{\mathbb{Q}}
\newcommand{\QQ}{\mathbb{Q}}
\newcommand{\Z}[0]{\mathbb{Z}}

\newcommand{\proj}[2]{\pi_{#1}(#2)}

\renewcommand{\vec}[1]{\mathbf{#1}}
\newcommand{\norm}[1]{\lVert#1\rVert}
\newcommand{\onenorm}[1]{\norm{#1}_1}
\newcommand{\infnorm}[1]{\norm{#1}_\infty}

\newcommand{\config}[2]{#1(#2)}
\newcommand{\Config}[2]{\config{#1}{\vec{#2}}}
\newcommand{\eff}{\textup{eff}}
\newcommand{\Counter}[1]{\texttt{#1}}
\newcommand{\pismall}{\pi_{\text{small}}}
\newcommand{\pitail}{\pi_{\text{tail}}}
\newcommand{\rhotail}{\rho_{\text{tail}}}
\newcommand{\rhosmall}{\rho_{\text{small}}}

\newcommand{\sumdim}{g}
\newcommand{\maxdim}{{g_{\textsc{scc}}}}
\newcommand{\ressumdim}{k_{\textsc{res}}} 
\newcommand{\cyclespace}[1]{\mathsf{CycleSpace}_\QQ(#1)}
\newcommand{\cycleeffects}{\mathsf{Cycle}}

\newcommand{\class}[1]{\textsc{#1}}
\newcommand{\NL}{\class{NL}\xspace} 

\newcommand{\np}{\class{NP}\xspace}
\newcommand{\pspace}{\class{PSpace}\xspace}
\newcommand{\expspace}{\class{ExpSpace}\xspace}
\newcommand{\twoexpspace}{\class{2-ExpSpace}\xspace}
\newcommand{\tower}{\class{Tower}\xspace}

\newcommand{\Ff}{\mathcal{F}}
\newcommand{\F}{\Ff}

\newcommand{\Tower}{\mathop{\operatorname{tower}}}

\newcommand{\cone}{\mathop{\operatorname{cone}}\nolimits}
\newcommand{\intcone}{\mathop{\operatorname{cone}}\nolimits_\N}

\newcommand{\trans}[1]{\xrightarrow{#1}}
\newcommand{\tran}{\rightarrow}

\EventEditors{John Q. Open and Joan R. Access}
\EventNoEds{2}
\EventLongTitle{42nd Conference on Very Important Topics (CVIT 2016)}
\EventShortTitle{CVIT 2016}
\EventAcronym{CVIT}
\EventYear{2016}
\EventDate{December 24--27, 2016}
\EventLocation{Little Whinging, United Kingdom}
\EventLogo{}
\SeriesVolume{42}
\ArticleNo{23}

\hideLIPIcs

\begin{document}

\maketitle

\begin{abstract}
The geometric dimension $g$ of a Vector Addition System with States (VASS) is the dimension of the vector space generated by cycles in the VASS; this parameter refines the standard dimension $d$, the number of counters.
Recently, it was discovered that the fastest-known algorithm for solving the reachability problem for VASS has the same complexity in terms of $g$ as in terms of $d$.
This suggests that the geometric dimension may in fact be a more adequate parameter for measuring the complexity of VASS reachability problems.
We initiate a more systematic study of the geometric dimension.
We discuss differences between two parameters: the geometric dimension and the SCC dimension.
Our main technical result states that classical results about the coverability and boundedness problems can be improved from dimension $d$ to geometric dimension $g$.
Namely, coverability is witnessed by runs of length $n^{2^{\mathcal{O}(g)}}$ instead of $n^{2^{\mathcal{O}(d)}}$, and unboundedness can be witnessed by runs of length $n^{2^{\mathcal{O}(g\log g)}}$ instead of $n^{2^{\mathcal{O}(d\log d )}}$, where $n$ is the size of the instance.
We also study integer reachability and simultaneous unboundedness in VASS parameterised by the geometric dimension. 

\end{abstract}



\section{Introduction}
\label{sec:introduction}
Vector Addition Systems with States (VASS), or equivalently Petri nets, are one of the most popular models of concurrent systems with many theoretical and practical applications~\cite{EsparzaN94,Reisig13,siglog/Schmitz16}.
VASS can be seen as automata with several nonnegative integer counters, that can be incremented and decremented by transitions, but not tested for equality with zero. Such automata with $d$ counters are called $d$-dimensional VASS (shortly $d$-VASS).
Investigating algorithmic problems for VASS is an active research field since the '70s. 
Reachability is a central problem which asks, for a given VASS and two configurations (states together with counter values), whether there exists a run from the source configuration to the target configuration. 
Other significant problems include coverability and boundedness.
The coverability problem asks whether there is a run from the source configuration to some configuration above the target.
The boundedness problem asks whether the set of configurations reachable from the source is finite. 
Already in 1976, coverability was shown to be \expspace-hard~\cite{Lipton76}.
This immediately implies that reachability and boundedness are also \expspace-hard (due to trivial reductions).
In 1978, Rackoff proved that both coverability and boundedness are in \expspace~\cite{Rackoff78}.
Three years later, reachability was proved to be decidable~\cite{Mayr81}. 
For many subsequent decades, determining the complexity of the reachability problem was a major open problem.
Despite plenty of effort, the Ackermannian complexity was only settled in 2021~\cite{LerouxS19,CzerwinskiO21,Leroux21}.

Contemporary research efforts have been devoted to studying reachability, coverability, and boundedness with seemingly the most natural VASS parameter: the dimension $d$ (i.e.\ the number of counters).
Examples of this fact are abundant; works on reachability in $1$-VASS~\cite{ValiantP75,HaaseKOW09}, $2$-VASS~\cite{BlondinEFGHLMT21}, $3$-VASS~\cite{CzerwinskiJLO25}, $d$-VASS for $d \in \set{4, \ldots, 8}$~\cite{Czerwinski0LLM20,CzerwinskiO22}, for arbitrary $d$-VASS~\cite{LerouxS19,Leroux21,CzerwinskiO21,Lasota22,CzerwinskiJLLO23}, and even works on reachability in subclasses of VASS~\cite{Leroux21Flat,ChistikovCMOSW24} use the dimension $d$ as a parameter. 
Investigation of the coverability problem~\cite{Rackoff78,RosierY86,MazowieckiSW23,KunnemannMSSW23,PilipczukSS25} and reachability problems in VASS extended with a pushdown stack~\cite{LerouxST15,PerezR24,BiziereC25} or a branching mechanism~\cite{LazicS15,GollerHLT16,FigueiraLLMS17,MazowieckiP19,BiziereHLS25} have predominantly parameterised the VASS in question by the dimension $d$.

The papers that prove the Ackermannian upper bound of reachability use the dimension $d$ as the parameter: Leroux and Schmitz proved that reachability in $d$-VASS is in \(\F_{d+4}\)~\cite{LerouxS19}, then Fu, Yang, and Zheng later improved this to \(\F_{d}\)~\cite{FuYZ24, FuZY25}.
Here, $\F_d$ is the $d$-th level of the fast growing hierarchy of complexity classes~\cite{Schmitz16}.
Recently however, it was observed that the number of counters $d$ is not always the most appropriate parameter for analysing the complexity of reachability problems in VASS. 
Fu, Yang, and Zheng observed that \emph{both} the $\Ff_{d+4}$ and the $\Ff_d$ upper bounds also hold when the VASS is instead parameterised by the \emph{geometric dimension} \(d\). 
The geometric dimension of a Strongly Connected Component (SCC) of a VASS is defined to be the dimension of the vector space spanned by the effects of cycles in the SCC. 
The canonical version of extending this notion to the entire VASS is to take the (Minkowski) sum of the cycle spaces of each SCC.
Fu, Yang, and Zheng used this definition of the geometric dimension, denoted \(\sumdim\), to parameterise VASS~\cite{FuYZ24, FuZY25}.
It is possible however, to take the maximum over the geometric dimensions of the SCCs in the VASS; we call this parameter the \emph{SCC dimension}, denoted \(\maxdim\).
This parameter has been used recently by Guttenberg, Czerwi\'{n}ski, and Lasota when studing reachability problems in VASS extended with nested zero tests (a certain kind of restricted zero-testing ability)~\cite{GuttenbergCL25}.

In fact, the reachability algorithm runs in \(\F_d\) for geometric dimension \(d\) in both variants (the geometric dimension and the SCC dimension).
Intuitively, the algorithm works in the same complexity if we add new counters which are copies, rescalings, or even linear combinations of existing counters. 
It even works if the copied counter changes between every SCC (for example, if counter \(6\) first copies counter \(3\) and then in the next SCC copies counter \(1\)). Notice however, that the coefficients of these linear combinations can be negative and need not be integer, so the VASS with the additional counters may have significantly fewer runs. 
To illustrate the difference between the geometric dimension $\sumdim$ and the SCC dimension $\maxdim$, consider Linear Path Schemes (LPS)~\cite{LerouxS04,BlondinEFGHLMT21,ChistikovCMOSW24}, a well-known subclass of VASS. 
LPS are characterised by their structure: a (simple) path which has disjoint simple cycles at each node along the path.
Observe that LPS can have arbitrary geometric dimension $\sumdim \in \NN$, but can only have SCC dimension one $\maxdim \leq 1$.
This is because every SCC contains just one cycle but these cycles need not have pairwise linearly dependent effects.

\paragraph*{Our Contributions}
In this paper we initiate a more systematic study of two geometric VASS parameters: the geometric dimension and the SCC dimension.
Our focus is on the set of classical and more tractable reachability problems in VASS: coverability, simultaneous unboundedness, boundedness, and integer reachability.
Algorithms for solving these problems are often used as subroutines in algorithms for solving other problems.
We also discuss differences between geometric dimension and SCC dimension taking into account the reachability problem.

Our main technical contribution is a generalisation of the results of~\cite{Rackoff78} and the more recent~\cite{KunnemannMSSW23} from VASS parameterised by the dimension $d$ to VASS parameterised by the geometric dimension $\sumdim$. 
More concretely, Rackoff proved two seminal results for $d$-VASS of size $n$~\cite{Rackoff78}. 
The first is that if there is a run from the source $s$ that covers the target $t$, then there is also a run of length at most $\max(\size(t),n)^{2^{\Oo(d \log(d))}}$ from the source $s$ that covers the target $t$~{\cite[Section 3]{Rackoff78}}.
The second is that if the set of configurations reachable from the source configuration $s$ is infinite, then it is witnessed by a ``self-covering run'' of length at most $n^{2^{\Oo(d \log(d))}}$~{\cite[Section 4]{Rackoff78}}.
A self-covering run is one from the source $s$, that first reaches a configuration $c$ and then goes on to reach a configuration $c'$ that strictly greater than $c$.
Recently, K\"{u}nnemann, Mazowiecki, Sch\"{u}tze, Sinclair-Banks, and W\k{e}grzycki improved on Rackoff's bound for the length of coverability witnesses and showed that if there is a run from the source $s$ that covers the target $t$, then there is a run from $s$ that covers $t$ of length at most $\max(\size(t),n)^{2^{\Oo(d)}}$~\cite{KunnemannMSSW23}. The difference being that the exponent was decreased from $2^{\Oo(d \log(d))}$ to $2^{\Oo(d)}$ which matches Lipton's lower bound on the length of coverability witnesses~\cite{Lipton76}. 
We generalise all three of these results to geometric dimension \(g\) (\cref{thm:coverability-sum-dim},~\cref{thm:boundedness}, and~\cref{thm:improved-coverability-sum-dim}).
Our main conceptual insights are novel definitions of ``clean basis'' (\cref{def:clean-basis}) and ``geometrically small configurations'' (\cref{def:small-configurations}).

With the aim of delivering a well-designed toolbox for VASS with fixed geometric dimension, we generalise one more classic result of a flavour similar to techniques used by Rackoff in~\cite{Rackoff78}. 
A necessary piece in the standard approach (called the KLM-decomposition) to solving reachability in VASS was the following observation about the situation in which all counters are simultaneously unbounded. 
If, from a given source configuration $s$, one cannot cover \emph{all} configurations in some state $q$, then there exists a bound $C$ such that, on every run starting from $s$, there is at least one counter whose value is bounded above by $C$.
Importantly, no counter is guaranteed to be bounded in all configurations reachable from $s$, the bounded counter may depend on the run.
Originally, the constant $C$ was nonprimitive-recursive~\cite{Kosaraju82}, as it was taken from the construction of the Karp-Miller tree~\cite{KarpM69}, but Leroux and Schmitz observed in~\cite{LerouxS19} that $C$ can be doubly-exponential (see also~{\cite[Lemma 31]{CzerwinskiJLO25arxiv}} for an explicit statement of this result), which significantly lowered the complexity of their algorithm. 
We generalise this result to geometric dimension $\sumdim$ (\cref{thm:no-pump}).

Furthermore, we consider the integer reachability problem for VASS which asks whether a given target configuration can be reached from a given source configuration by a run in which counters may have negative values. 
It is folklore (see~\cite{HaaseH14}) that if there is an integer run from the source $s$ to the target $t$ in a $d$-VASS, for a fixed $d$, then there is a polynomial  length integer run from $s$ to $t$ (with respect to the size of the VASS encoded in unary). 
We generalise this result to geometric dimension $\sumdim$ (\cref{thm:z-vass-sum}).

Finally, we conclude by comparing the two VASS parameters in question: geometric dimension and SCC dimension. 
For the geometric dimension, we have managed to show that for the classical problems coverability, boundedness, integer reachability, and simultaneous unboundedness the state-of-the-art algorithms can be generalised from the dimension to the geometric dimension. 
It is known that sometimes geometric dimension behaves differently than the number of counters. 
For example, reachability in binary VASS with geometric dimension 1 is \pspace-hard~\cite{Zheng25}, in contrast to \np-completeness~\cite{HaaseKOW09} for binary 1-VASS.
By and large however, we do find that geometric dimension behaves very similarly to the number of counters.

The situation is different for the SCC dimension. 
The algorithm solving the reachability problem for $d$-VASS~\cite{LerouxS19} works equally well for SCC dimension $\maxdim = d$, but it seems to be the only case in which we are currently able to generalise existing results from the dimension to the SCC dimension.
As corollaries of our results when VASS are parameterised by the geometric dimension $\sumdim$, we obtain upper bounds when the VASS is instead parameterised by the SCC dimension.
Specifically, for coverability, boundedness, integer reachability, and simultaneous unboundedness we can use the inequality $\sumdim \leq \maxdim \cdot n$, where $n$ is the number of states in the VASS.
Notice however, that in contrast to fixed geometric dimension $\sumdim$, we do not get polynomial upper bounds on the size of coverability, unboundedness, or integer reachability witnesses for fixed SCC dimension $\maxdim$. 
We remark that even for fixed $\sumdim$, our polynomial upper bound on the witness size does not imply $\NL$ algorithms.
This is because the dimension $d$ is not fixed, so one cannot store a configuration in logarithmic space.

Despite the fact that the general VASS reachability algorithm in~\cite{LerouxS19} works fine with respect to the SCC dimension, it seems to be the case that reachability in VASS with fixed SCC dimension $\maxdim$ has higher complexity than reachability in VASS with the corresponding dimension.
For this, we provide the following pieces of evidence.

First, as shown in~\cite{FuYZ24},
in VASS with $\sumdim = 2$ the reachability relation is semilinear and can be represented by a finite union of LPSs, similarly as for $2$-VASS~\cite{BlondinEFGHLMT21}. 
For $2$-VASS, these Linear Path Schemes (LPSs) have at most exponential size and contain at most polynomially many cycles.
This implies exponential length shortest runs. However, the construction described in~\cite{FuYZ24} does not provide any bound on the size of the constructed LPSs. It seems that the obtained LPSs are of at most doubly-exponential size (due to a polynomial blowup for each increment of the dimension), which would give at most triply-exponential run length, but there is no such claim in~\cite{FuYZ24}. 
That means that the currently existing techniques could possibly be used to show that the shortest runs witnessing reachability in VASS with SCC dimension $\maxdim = 2$ have triple exponential length.
In turn, this would yield an \twoexpspace upper bound for the reachability problem, in contrast to the \pspace upper bound for reachability in $2$-VASS~\cite{BlondinEFGHLMT21}.

For $\maxdim = 3$, the situation also seems more complicated than for $3$-VASS. We observe that an example of a binary $4$-VASS presented in~{\cite[Section 5]{Czerwinski0LLM20}} has SCC dimension $\maxdim = 3$. 
In this example, the shortest run between the source and target configuration has double exponential length. 
Such an example is not known for $3$-VASS; all known examples have the property that the shortest run between source and target is has at most exponential length. 
Therefore, the naive \pspace algorithm that enumerates all the possible exponential length runs is insufficient VASS with SCC dimension $\maxdim = 3$, but such an algorithm can be conjectured to solve reachability in $3$-VASS.

Last but not least we consider $\maxdim = 4$.
We prove that the reachability problem for VASS with SCC dimension $\maxdim = 4$ is \tower-hard (\cref{thm:max-dim-4-tower-hard}). 
To contrast it with $d$-VASS, the lowest number of counters for which reachability is known to be \tower-hard is $d = 8$~\cite{CzerwinskiO22}. 
Our construction (to prove~\cref{thm:max-dim-4-tower-hard}) follows the same ideas used to obtain \tower-hardness for reachability in $8$-VASS presented in~\cite{CzerwinskiO22}. 
This indicates that VASS of small SCC dimension can be much more expressive than VASS with a small number of counters.

\paragraph*{Paper Organisation} 
In Section~\ref{sec:preliminaries} we define preliminary notions and recall useful facts. Next Sections~\ref{sec:coverability},~\ref{sec:simultaneous-unboundedness},~\ref{sec:boundedness},~and~\ref{sec:z-reachability}, contain our main technical results about various algorithmic problems for VASS parameterised by geometric dimension and SCC dimension.
In Section~\ref{sec:coverability}, we present our results about the coverability problem. 
Next, we consider simultaneous unboundedness in Section~\ref{sec:simultaneous-unboundedness}. Section~\ref{sec:boundedness} is devoted to the boundedness problem. 
In Section~\ref{sec:z-reachability}, we study the integer reachability problem.
In Section~\ref{sec:tower-max-dim-4-lower-bound}, we prove \tower-hardness of the reachability problem for VASS with $\maxdim = 4$. Finally, in Section~\ref{sec:future}, we conclude and present future research directions.

\section{Preliminaries}
\label{sec:preliminaries}
We use $\mathbb{Z}, \mathbb{N}, \mathbb{Q}$ and $\mathbb{Q}_{\ge 0}$ to denote the sets of integers, nonnegative-integers, rational numbers, and nonnegative rational numbers, respectively. Let $d \in \mathbb{N}$ be a number. We denote $[d]$ for the set $\{1, 2, \ldots, d\}$. In particular $[0] = \emptyset$ is the empty set. Given a $d$-dimensional vector $\vec{x} \in \mathbb{Q}^d$, we write $\vec{x}[i]$ for its $i$-th component, where $i \in [d]$. Hence $\vec{x} = (\vec{x}[1], \vec{x}[2], \ldots, \vec{x}[d])$. For $q \in \Q$ we denote $\vec{q}^d = (q, q, \ldots, q) \in \mathbb{Q}^d$, and write $\vec{e}_i^d$ for the $d$-dimensional vector satisfying $\vec{e}_i^d[i] = 1$ and $\vec{e}_i^d[j] = 0$ for all $j \ne i$. When the dimension $d$ is clear from the context, we shall drop the super-script $d$ and simply write $\vec{q}$, $\vec{e}_i$ for these vectors. Vectors are compared in the component-wise manner. Given $\vec{x}, \vec{y} \in \mathbb{Q}^d$, we have $\vec{x} \le \vec{y}$ if $\vec{x}[i] \le \vec{y}[i]$ for all $i \in [d]$. We write $\vec{x} < \vec{y}$ if $\vec{x} \le \vec{y}$ and $\vec{x} \ne \vec{y}$. For a vector $\vec{x} \in \mathbb{Q}^d$, we will use its max-norm $\infnorm{\vec{x}} = \max_{i \in [d]} \abs{\vec{x}[i]}$ and its one-norm $\onenorm{\vec{x}} = \sum_{i = 1}^d \abs{\vec{x}[i]}$. For a set of vectors $X \subseteq \Q^d$ we write $\infnorm{T} = \max_{\vec{t} \in T} \infnorm{\vec{t}}$. For $\vec{x} \in \mathbb{Q}^d$ and $K \subseteq [d]$ we use the notation $\proj{K}{\vec{x}}$ for the projection of $\vec{x}$ onto indices in $K$.

\paragraph*{Vector Spaces and Cones}

Let $X \subseteq \mathbb{Q}^d$ be a finite set of vectors and let $\mathbb{S} \subseteq \mathbb{Q}$ be a set of numbers. The $\mathbb{S}$-linear combination of $X$ is the set $\mathbb{S}(X) = \{ \sum_{i = 1}^m \lambda_i \vec{m}_i \mid m \in \mathbb{N}, \lambda_i \in \mathbb{S}, \vec{m}_i \in X \}$. The \emph{vector space} generated by $X$ is defined to be $\mathbb{Q}(X)$. The \emph{(rational) cone} generated by $X$ is $\cone(X) = \mathbb{Q}_{\ge 0}(X)$. We also consider the \emph{integer cone} generated by $X$, which is $\intcone(X) = \mathbb{N}(X)$. Carath\'eodory's Theorem is a useful tool for cones.

\begin{theorem}[{Carath\'eodory's Theorem for Rational Cones, see e.g.\ \cite[Corollary 7.1i]{schrijver1998theory}}]
    \label{thm:caratheodory-rational}
    Let $X \subseteq \mathbb{Q}^d$ be a finite set of vectors and let $\vec{b} \in \cone(X)$. Then there exists a subset $\widehat{X} \subseteq X$ such that $\vec{b} \in \cone(\widehat{X})$ and $|\widehat{X}| \le d$.
\end{theorem}

We also need the analogy of Carath\'eodory's Theorem for integer cones.

\begin{theorem}[{Carath\'eodory's Theorem for Integer Cones, \cite[Theorem 1]{EisenbrandS06}}]
    \label{thm:caratheodory-integer}
    Let $X \subseteq \mathbb{Z}^d$ be a finite set of vectors and let $\vec{b} \in \intcone(X)$. Then there exists a subset $\widehat{X} \subseteq X$ such that $\vec{b} \in \intcone(\widehat{X})$ and $|\widehat{X}| \leq 2d \log(4dM)$ where $M = \max_{x \in X} \infnorm{\vec{x}}$. 
\end{theorem}

\paragraph*{Vector Addition Systems with States}

A \emph{$d$-dimensional vector addition system with states} ($d$-VASS) is a pair $V = (Q, T)$ where $Q$ is a finite set of \emph{states} and $T \subseteq Q \times \mathbb{Z}^d \times Q$ is the set of \emph{transitions}. A transition of the form $(p, \vec{a}, q)$ may be written as $p \xrightarrow{\vec{a}} q$ for clarity, and the vector $\vec{a}$ is called the \emph{effect} of this transition. 
A VASS may be viewed as a directed graph where edges are labeled with integer vectors. 
Accordingly, a \emph{path} in a VASS is a sequence of transitions and a \emph{cycle} is a path that starts and ends at the same state.
Given a path (or a cycle) $\pi$ of $V$, its effect $\eff(\pi)$ is the sum of effects of all transitions constituting $\pi$. The (unary-encoded) size of $V$ is defined by $\size(V) = |Q| + |T| \cdot d \cdot (\max_{u \in T}\infnorm{\eff(u)} + 1)$. The max-norm of $V$ is defined to be $\infnorm{T} := \infnorm{\{\eff(u) \mid u \in T\}}$.

A configuration of $V$ consists of a state $q \in Q$ and a nonnegative vector $\vec{x} \in \mathbb{N}^d$. 
Such a configuration shall be written as $q(\vec{x})$.
For a configuration $c = \Config{q}{x}$, its \emph{norm} is defined by $\infnorm{c} = \infnorm{\vec{x}}$ and $\onenorm{c} = \onenorm{\vec{x}}$. We write $c[i]$ for $\vec{x}[i]$ where $i \in [d]$. We remark that coordinates of the vector $\vec{x}$ will often be called ``counters'', as VASS is a special form of counter automata. Configurations can be compared when they are in the same state. For $c = q(\vec{x})$ and $c' = q'(\vec{x}')$, we write $c \le c'$ if $q = q'$ and $\vec{x} \le \vec{x}'$, and $c < c'$ if $q = q'$ and $\vec{x} < \vec{x}'$.

Each transition $u = (p, \vec{a}, q)$ induces \emph{one-step runs} $p(\vec{x}) \xrightarrow{t} q(\vec{y})$ between configurations $p(\vec{x}), q(\vec{y})$ satisfying $\vec{x} + \vec{a} = \vec{y}$. Similarly, a path $\pi = u_1 u_2 \ldots u_k$ induces \emph{runs} of the form $c_0 \xrightarrow{u_1} c_1 \xrightarrow{u_2} c_2 \cdots \xrightarrow{u_k} c_k$. We emphasize that vectors in each configuration $c_i$ along the run should be nonnegative. We write $c \xrightarrow{*}_V c'$ if there is a run from $c$ to $c'$ in $V$. The subscript $V$ is often omitted when the VASS $V$ is clear from the context. We say a run from $c$ to $c'$ \emph{covers} a configuration $c''$, if it holds that $c' \ge c''$.

We also consider \emph{$\Z$-configurations} $p(\vec{z})$ where the vector $\vec{z} \in \Z^d$ is allowed to have negative values. Runs over $\Z$-configurations are defined similarly, and are called \emph{$\Z$-runs}.

Problems that are commonly studied on VASS include:

\begin{description}
    \item[Reachability.] Given a VASS $V$ and configurations $c, c'$, is there a run from $c$ to $c'$ in $V$?
    \item[\(\Z\)-Reachability.] Given a VASS $V$ and $\Z$-configurations $c, c'$, is there a $\Z$-run from $c$ to $c'$?
    \item[Coverability.] Given a VASS $V$ and configurations $c, c'$, is there a run from $c$ that covers $c'$?
    \item[Boundedness.] Given a VASS $V$ and a configuration $c$, is the set $\{c' \mid c \xrightarrow{*} c'\}$ finite?
\end{description}

In this paper we focus on these problems under the parameterisation by \emph{geometric dimension} which we introduce next.

\paragraph*{Geometric Dimension of a VASS}

Let $V = (Q, T)$ be a $d$-VASS and $q \in Q$ be a state. We denote by $\cycleeffects_V(q) = \{\eff(\theta) \mid \theta \text{ is a cycle containing }q\}$ the set of effects of all cycles containing the state $q$. Two types of geometric dimension have been proposed in the literature.
The \emph{geometric dimension} $\sumdim(V)$ and its variant the \emph{SCC dimension} $\maxdim(V)$ are respectively defined by 
\begin{align*}
    & \sumdim(V) = \dim(\QQ(\bigcup_{q \in Q} \cycleeffects_V(q))) && \maxdim(V) = \max_{q \in Q} \dim(\QQ(\cycleeffects_V(q)))
\end{align*}

In particular, we define the \emph{cycle space} of $V$ as $\cyclespace{V} = \QQ(\bigcup_{q \in Q} \cycleeffects_V(q))$, which is the vector space spanned by the effects of all cycles in $V$. Hence, the geometric dimension is indeed the dimension of the cycle space of $V$. We remark that $\cyclespace{V}$ can be generated from effects of $\sumdim(V)$ many simple cycles.

Observe that the space $\QQ(\cycleeffects_V(q))$ only depends on the strongly connected component (SCC) in $V$ that contains $q$, hence the maximum can be taken over SCCs instead of states as indicated by the name SCC dimension.  From the definitions we observe the following simple fact on the relation between the two types of geometric dimension:

\begin{lemma}
    \label{lem:sum-dim-le-max-dim-mult-num-states}
    Let $V = (Q, T)$ be a VASS. Then $\sumdim(V) \le |Q| \cdot \maxdim(V)$.
\end{lemma}

\begin{proof}
    For each $q \in Q$, the vector space $\QQ(\cycleeffects_V(q))$ is spanned by at most $\maxdim(V)$ generators. Therefore, the cycle space of $V$ has no more than $|Q| \cdot \maxdim(V)$ generators.
\end{proof}

\paragraph*{Integer Solutions of Linear Systems} 

The following result on system of linear equations will be useful.

\begin{lemma}[{\cite[Prop.\ 4]{linear_equations}}]
    \label{lem:linear_equations}
    Consider a system $A \cdot \vec{x} = \vec{b}$ of $d$ Diophantine linear equations with $n$ unknowns, where absolute values of coefficients in $A$ and $\vec{b}$ are bounded by $N$. If the system has a nonnegative integer solution, then it also has one bounded by $\infnorm{\vec{x}} \le \oo(nN)^d$.
\end{lemma}

Similar bounds also exist for systems of linear inequalities, see e.g.\ \cite[Prop.\ 3]{linear_equations}. In this paper we need a stronger bound that better fits the geometric dimension, in the sense that the exponent depends on the rank of the matrix $A$ instead of the number of variables.

\begin{restatable}{lemma}{clmLinearInequalities}
    \label{clm:linear_inequalities}
    Consider a system $A \cdot \vec{x} \geq \vec{b}$ of $d$ Diophantine linear inequalities with $n$ unknowns, where absolute values of coefficitents of $A$ and $\vec{b}$ are bounded by $N$ and, moreover, $\vec{b} \geq \vec{0}$. If the system has a nonnegative integer solution then it has one of max-norm bounded by $\infnorm{\vec{x}} \le (r+1)(rN)^{r}$ where $r = \rank(A)$.
\end{restatable}

Proof of \cref{clm:linear_inequalities} can be found in \cref{sec:app_preliminaries}. We remark that it is based on \cite{ilp}.

\section{Coverability}
\label{sec:coverability}
We will focus on the coverability problem for VASS parameterised by geometric dimension $\sumdim$.
The main result of this section is that coverability is witnessed by doubly-exponential length runs, where the double exponential dependence is only on $\sumdim$, rather than the size or (standard) dimension of the VASS (\cref{thm:coverability-sum-dim}).
This means that coverability in VASS with fixed geometric dimension is witnessed by polynomial length runs.

Throughout this section, unless otherwise specified, we fix our attention on an instance of coverability in a VASS $V = (Q, T)$, an initial configuration $\Config{s}{x}$, and a target configuration $\Config{t}{y}$.
We also fix the following parameters of $V$:
\begin{itemize}
        \item $d \in \NN$ is the dimension,
        \item $\sumdim \leq d$ is the geometric dimension, 
        \item $n \coloneqq \abs{Q}$ is the number of states,
        \item $M \coloneqq \max\set{\infnorm{\vec{x}} : (p, \vec{x},q) \in T}$ is the magnitude of the greatest effect of any transition.
\end{itemize}

We can assume, without loss of generality that $V$ consists of a series of strongly connected components (SCC) that are connected by single transitions.
This is because a run in $V$ from $\Config{s}{x}$ that covers $\Config{t}{y}$ follows a path that traverses a series of SCCs.

\begin{restatable}{theorem}{coverabilitySumDim}\label{thm:coverability-sum-dim}
    Let $f: \NN \to \NN$, $f(x) \coloneqq (x+1)^{x+1}$.
    If there is a run in $V$ from $\Config{s}{x}$ that covers $\Config{t}{y}$, then there exists a run in $V$ from $\Config{s}{x}$ that covers $\Config{t}{y}$ of length at most $\poly(d,n,M,\infnorm{\vec{y}})^{f(\sumdim)}$.
\end{restatable}

Our proof of~\cref{thm:coverability-sum-dim} has the same structure as Rackoff's original proof~{\cite[Section 3]{Rackoff78}} that coverability in VASS is witnessed by double exponential length runs.
Roughly speaking, the proof goes as follows.
One observes that if a counter becomes ``sufficiently large'', then there is no need to worry about that counter dropping below zero for the remainder of the run.
Accordingly, the proof goes by induction on the dimension (at each step, another counter is ``ignored'').
One of the core components in the proof is the definition of ``small configurations'', where a configuration $\Config{q}{v}$ that is not small has a large counter value $\vec{v}[i]$.

The proof of~\cref{thm:coverability-sum-dim} follows suit, though the definition of what makes a configuration small is noticeably different.
Here is where we make our main technical contribution; we define ``geometrically small configurations'' (\cref{def:small-configurations}).
Towards proving~\cref{thm:coverability-sum-dim}, after defining geometrically small configurations, we will first state and prove some intermediate facts that will help us then prove~\cref{lem:rackoff} (which is the core lemma behind~\cref{thm:coverability-sum-dim}).

\begin{definition}[Clean basis]
\label{def:clean-basis}
    Let $U \sset \QQ^d$ be a $g$-dimensional vector space.
    We say that a basis $B = \set{\vec{b}_1, \ldots, \vec{b}_g}$ of $U$ is \emph{clean} if there exist coordinates $K = \set{k_1, \ldots, k_g} \sset [d]$ such that $\proj{K}{\vec{b}_1}, \ldots, \proj{K}{\vec{b}_g}$ forms the $g$-dimensional identity matrix. Coordinates $K$ are called the \emph{distinguished coordinates} of $B$.
\end{definition}

Observe that any basis $B$ of a $g$-dimensional vector space $U \sset \QQ^d$ can be turned into a clean basis using a Gaussian-elimination-style algorithm. 

\begin{definition}[Geometrically small configurations]
\label{def:small-configurations}
    A vector $\vec{v} \in \N^d$ is \emph{$C$-small} with respect to $g$-dimensional vector space $U \subseteq \Q^d$ if there exists a clean basis $B = \set{\vec{b}_1, \ldots, \vec{b}_g}$ of $U$ with the set $K \subseteq [d]$ of $g$ distinguished coordinates, such that $\vec{v}[i] < C$ for all $i \in K$.
    Moreover, we say that a configuration $\Config{q}{v}$ of a VASS $V$ is \emph{$C$-small} if $\vec{v}$ is $C$-small with respect to $\cyclespace{V}$.
    If a vector $\vec{v}$ or a configuration $\Config{q}{v}$ is not $C$-small then we call it \emph{$C$-large}.
\end{definition}

Now, similar to Rackoff's original proof, we would like to bound the number of $C$-small configurations.
Unfortunately, there can be an infinite number of $C$-small configurations as there are no restrictions on the coordinates that are not distinguished.
Instead, we will bound the number of $C$-small configurations that are \emph{reachable} from some initial configuration (\cref{clm:small-configurations}).

Before stating~\cref{clm:small-configurations}, we introduce a way of decomposing counter value vectors that will be useful for~\cref{clm:small-configurations} and later.
We can describe a vector $\vec{v}$ of counter values of a configuration $\Config{q}{v}$ as the sum $\vec{v} = \vec{x} + \vec{r}_q + \vec{w}$, where $\vec{x}$ are the starting counter values, $\vec{r}_q$ is the effect of some path from $s$ to $q$, and what remains is $\vec{w}$ that belongs to $\cyclespace{V}$.
In order to meaningfully use this decomposition, we need to argue that we can arbitrarily nominate a path from $s$ to $q$ which allows us to fix $\vec{r}_q$, for every state $q$.

\begin{restatable}{claim}{pathDifference}\label{clm:path-difference}
    Let $p, q \in Q$ be two states and let $\rho_1, \rho_2$ be two paths from $p$ to $q$.
    The difference between the effects of $\rho_1$ and $\rho_2$ belongs to $\cyclespace{V}$.
\end{restatable}

\begin{proof}[Proof Sketch]
    First, consider the case when $p$ and $q$ are in the same strongly connected component.
    Let $\tau$ be any path from $q$ back to $p$.
    Then $\rho_1 \tau$ is a cycle from $p$ to $p$, and so is $\rho_2 \tau$. 
    Thus the difference of the effects of $\rho_1 \tau$ and $\rho_2 \tau$ belongs to $\cyclespace{V}$.
    Clearly, the difference in effect between \(\rho_1 \tau\) and \(\rho_2 \tau\) is the same as between the effects of $\rho_1$ and $\rho_2$.
    
    Since \(V\) is a line of SCCs, \(\rho_1\) and \(\rho_2\) visit the same SCCs. We use the above case for the segments of \(\rho_1\) and \(\rho_2\) in the respective SCCs, for the computation see~\cref{app:coverability}.
\end{proof}

\begin{restatable}{claim}{smallConfigurations}\label{clm:small-configurations}
    For every configuration $\Config{p}{u}$ and every $C \in \N$, there are at most $n \cdot (d \cdot C)^\sumdim$ many $C$-small configurations reachable from $\Config{p}{u}$.
\end{restatable}

\begin{proof}[Proof Idea]
    The idea is to prove two facts: 1) For every set \(K \subseteq [d]\) of distinguished coordinates, there is at most one clean basis \(B\). 2) If \(\Config{q}{v}, \Config{q}{v'}\) are both reachable from \(\Config{p}{u}\), and \(\proj{K}{\vec{v}}=\proj{K}{\vec{v}'}\), then \(\vec{v}=\vec{v}'\) (here we use~\cref{clm:path-difference}).
    
    This leads to the given bound of \(|Q| \cdot \binom{d}{g} \cdot C^{\sumdim} \leq n \cdot (d \cdot C)^{\sumdim}\) since any \(C\)-small configuration is uniquely determined by its state \(q \in Q\), the subset \(K \subseteq [d]\) s.t. it is small w.r.t. the corresponding clean basis \(B\) and the vector \(\proj{K}{\vec{v}} \in [0, C-1]^{\sumdim}\).
\end{proof}

Now, again similar to Rackoff's original proof, we would like to argue that if a configuration is $C$-large, then there is some counter that is sufficiently large (so it can be ``ignored'').
We aim to prove something stronger.
Namely, if a configuration is $C$-large, then there is a collection of counters which are all sufficiently large and by ``ignoring'' them, the geometric dimension of the VASS decreases.
The following claim states that if a vector $\vec{v}$ is $C$-large, then there is some clean basis $B$ and a vector from this basis $\vec{b}_i$ which can be used to identify which coordinates of $\vec{v}$ are large.

\begin{claim} \label{clm:c-large}
    Let $U \sset \QQ^d$ be a $g$-dimensional vector space, let $\vec{v} \in \N^d$, and let $C \in \NN$.
    If $\vec{v}$ is $C$-large with respect to $U$, then there exists a clean basis $B = \set{\vec{b}_1, \ldots, \vec{b}_g}$ of $U$ and an index $i \in [g]$ such that, for every $j \in [d]$, $\vec{b}_i[j] \neq 0$ implies that $\vec{v}[j] \geq C$. 
\end{claim}
\begin{proof}
    Assume, for sake of contradiction, that $\vec{v}$ is $C$-large and there does not exist a clean basis $B$ and an index $i \in [g]$ such that, for every $j \in [d]$, $\vec{b}_i[j] \neq 0$ implies that $\vec{v}[j] \geq C$.
    Consider a clean basis $B = \set{\vec{b}_1, \ldots, \vec{b}_g}$ with a distinguished set of coordinates $K = \set{k_1, \ldots, k_g}$ for which $N \coloneqq \onenorm{\proj{K}{\vec{v}}}$ is minimised. 
    
    Now, our goal is to construct a new clean basis $B'$ with a distinguished set of coordinates $K'$ such that $\onenorm{\proj{K'}{\vec{v}}}$ is less than $N$.
    Since $\vec{v}$ is $C$-large, there exists $i \in [g]$ such that $\vec{v}[k_i] \geq C$.
    By assumption, there exists $j \in [d]$ such that $\vec{b}_i[j] \neq 0$ and $\vec{v}[j] < C$. 
    Now, we will define $B' = \set{\vec{b}'_1, \ldots, \vec{b}'_g}$ by
    \begin{equation*}
        \vec{b}'_i \coloneqq \frac{\vec{b}_i}{\vec{b}_i[j]} 
        \text{\; and \;}
        \vec{b}'_k \coloneqq \vec{b}_k - \vec{b}_k[j]\cdot\vec{b}'_i, \text{ for all } k \neq i.
    \end{equation*}
    We identify the set of distinguished coordinates $K' = \set{k_1, \ldots, k_{i-1}, j, k_{i+1}, \ldots, k_g}$.
    
    We shall now verify that indeed $B'$ is a clean basis of $U$ with distinguished coordinates $K'$.
    First, notice that $\vec{b}'_i$ is defined to be equal to a linear scalar of $\vec{b}_i$.
    This means that the only non-zero coordinate of $\proj{K'}{\vec{b}'_i}$ is the $j$-th coordinate; and clearly $\vec{b}'_i[j] = 1$.
    Second, we will argue that the distinguished coordinate for $\vec{b}'_k$ is the same as the distinguished coordinate for $\vec{b}_k$.
    Suppose that $k \in K$ was the distinguished coordinate for $\vec{b}_k$ (i.e.\ $\proj{K}{\vec{b}_k} = (0, \ldots, 0, 1, 0, \ldots, 0) = \vec{e}_k$ is the $k$-th standard basis vector).
    We verify that $\proj{K'}{\vec{b}'_k} = \vec{e}_k$ as well.
    \begin{itemize}
        \item For the $k$-th coordinate: $\vec{b}'_k[k] = \vec{b}_k[k] - \vec{b}_k[j] \cdot \vec{b}'_i[k] = 1 - 0$.
        \item For the $j$-th coordinate: $\vec{b}'_k[j] = \vec{b}_k[j] - \vec{b}_k[j] \cdot \vec{b}'_i[j] = \vec{b}_k[j] - \vec{b}_k[j] = 0$.
        \item For the $k'$-th coordinate ($k' \in K' \setminus\set{k, j}$): $\vec{b}'_k[k'] = \vec{b}_k[k'] - \vec{b}_k[j]\cdot \vec{b}'_i[k'] = 0 - 0$.
    \end{itemize}
    
    Finally, to see that $\onenorm{\proj{K'}{\vec{v}}} < N$, observe that (i) $K'$ only differs from $K$ by the replacement of $k_i$ with $j$, (ii) $\vec{v}[k_i] \geq C$, and 
    (iii) $\vec{v}[j] < C$.
    Thus, $\onenorm{\proj{K}{\vec{v}}} = \onenorm{\proj{K'}{\vec{v}}} - \vec{v}[k_i] + \vec{v}[j] < \onenorm{\proj{K}{\vec{v}}} = N$ 
    (contradicts the minimality of $N$).
\end{proof}

With the definition of geometrically small configurations in hand and an understanding that configurations that are not small will have counters with large values, we are ready to prove~\cref{thm:coverability-sum-dim}.
As a first step towards proving~\cref{thm:coverability-sum-dim}, we shall define the bound on the length of runs that we would like to establish.
\begin{equation*}
    L_i \coloneqq
    \begin{cases}
        n - 1 & \text{ for } i = 0 \\
        n(d(\infnorm{\vec{y}} + M\cdot L_{i-1}))^i + L_{i-1} & \text{ for } i \geq 1
    \end{cases}
\end{equation*} 
Our goal is to prove that coverability is witnessed by runs of length at most $L_{\sumdim}$ (\cref{lem:rackoff}).

\begin{lemma} \label{lem:rackoff}
    If there is a run in $V$ from $\Config{s}{x}$ that covers $\Config{t}{y}$, then there is a run from $\Config{s}{x}$ that covers $\Config{t}{y}$ of length at most $L_{\sumdim}$.
\end{lemma}

\begin{proof}
    Proof by induction on the geometric dimension $\sumdim$.
    For $\sumdim = 0$, we know that all cycles in $V$ have zero effect.
    This means that the shortest run between any pair of configurations does not visit a state more than once.
    Hence, if there is a run in $V$ from $\Config{s}{x}$ that covers $\Config{t}{y}$, then there is one of length at most $L_0 = n-1$.
    
    For $\sumdim \geq 1$, we shall assume~\cref{lem:rackoff}
    holds for VASS with geometric dimension $\sumdim-1$.
    Now, assume that there is a run in $V$ from $\Config{s}{x}$ that covers $\Config{t}{y}$.
    Let $\Config{s}{x} \xrightarrow{\pi} \Config{t}{y'}$, for some $\vec{y}' \geq \vec{y}$, be the shortest run from $\Config{s}{x}$ that covers $\Config{t}{y}$.
    We will split this run about the first configuration that is $C$-large for $C = M\cdot L_{\sumdim-1} + \infnorm{\vec{y}}$.
    Precisely, consider the decomposition: $\Config{s}{x} \xrightarrow{\pismall} \Config{q}{m} \xrightarrow{\pitail} \Config{t}{y'}$, where $\Config{q}{m}$ is the first $C$-large configuration.
    
    Since $\Config{s}{x} \xrightarrow{\pi} \Config{t}{y'}$ has minimal length, we know that no configuration is repeated.
    This means that the length of $\pismall$ is at most $n \cdot (d \cdot C)^{\sumdim} = n(d(M\cdot L_{\sumdim-1} + \infnorm{\vec{y}}))^{\sumdim}$ (by~\cref{clm:small-configurations}).
    
    Now, we will use the inductive assumption to bound the length of $\pitail$.
    Since $\Config{q}{m}$ is $(M \cdot L_{\sumdim-1} + \infnorm{\vec{y}})$-large, by~\cref{clm:c-large}, there exists a clean basis $B = \set{\vec{b}_1, \ldots, \vec{b}_{\sumdim}}$ and an index $i \in [g]$ such that, for every $j \in [d]$, if $\vec{b}_i[j] \neq 0$, then $\vec{m}[j] \geq M\cdot L_{\sumdim-1} + \infnorm{\vec{y}}$.
    In fact, we shall define $J \coloneqq \set{j \in [d] : \vec{b}_i[j] \neq 0}$.
    Now, we shall remove the counters in $J$ from the VASS $V$.
    Let $V' = (Q, T')$ be the $(d - \abs{J})$-VASS obtained by projecting all transition vectors to the coordinates $[d] \setminus J$.
    Precisely, $T' = \set{(p, \proj{[d]\setminus J}{\vec{u}}, p') : (p, \vec{u}, p') \in T}$.
    Recall that $B = \set{\vec{b}_1, \ldots, \vec{b}_{\sumdim}}$ was a basis of $\cyclespace{V}$.
    We will argue that 
    \begin{equation*}
        B' = \set{\proj{[d]\setminus J}{\vec{b}_1}, \ldots, \proj{[d]\setminus J}{\vec{b}_{i-1}}, \proj{[d]\setminus J}{\vec{b}_{i+1}}, \ldots, \proj{[d]\setminus J}{\vec{b}_{\sumdim}}}
    \end{equation*}
    is a generating set of $\cyclespace{V'}$ which suffices to prove that the geometric dimension of $V'$ is $\sumdim-1$.
    Since the transition vectors of $V'$ are the transition vectors of $V$ projected to counters $[d] \setminus J$, we know that $\cyclespace{V'}$ is just $\cyclespace{V}$ projected to counters $[d] \setminus J$.
    This means that the span of $\set{\proj{[d]\setminus J}{\vec{b}_1}, \ldots, \proj{[d]\setminus J}{\vec{b}_{\sumdim}}}$ is $\cyclespace{V'}$.
    Now, notice that $\proj{[d]\setminus J}{\vec{b}_i} = \vec{0}$ (by definition of $J$).
    We can therefore remove $\proj{[d]\setminus J}{\vec{b}_i}$ from this set to obtain the generating set $B'$ of size $\sumdim - 1$.
    
    Now, we shall analyse the length of the suffix of the run in $V'$.
    Consider the configuration $q(\proj{[d]\setminus J}{\vec{m}})$.
    By the inductive assumption, we know that there is a run in $V'$ from $q(\proj{[d]\setminus J}{\vec{m}})$ that covers $t(\proj{[d]\setminus J}{\vec{y}})$ of length at most $L_{\sumdim-1}$.
    Now, we shall lift this run back to $V$ (by copying the appropriate transitions). 
    It remains to argue that the counters in $J$ remain above $0$ and cover the target $\vec{y}$.
    For this, recall that for all $j \in J$, $\vec{m}[j] \geq M \cdot L_{\sumdim-1} + \infnorm{\vec{y}}$.
    The most negative effect of a single transition is $-M$, thus over the suffix of length $L_{\sumdim-1}$, the greatest possible negative effect is $-M \cdot L_{\sumdim-1}$.
    This means that all counters in $J$ remain above zero, and at the end have counter values at least $\infnorm{\vec{y}}$ (which guarantees the target is covered).
    Thus, by following the same suffix from $\Config{q}{m}$ in $V$, we obtain a run in $V$ from $\Config{q}{m}$ that covers $\Config{t}{y}$ of length at most $L_{\sumdim-1}$.
    As $\pitail$ has minimal length, we therefore conclude that the length of $\pitail$ is at most $L_{\sumdim-1}$.
    
    Altogether, the length of $\pi$ is at most the length of $\pismall$ plus the length of $\pitail$ which is $n(d(M\cdot L_{\sumdim-1} + \infnorm{\vec{y}}))^{\sumdim} + L_{\sumdim-1} = L_\sumdim$ (as required).
\end{proof}

Now, to conclude this section by proving~\cref{thm:coverability-sum-dim}, we only need to argue that $L_{\sumdim}$ is indeed bounded by $\poly(d,n,M,\infnorm{\vec{y}})^{f(\sumdim)}$.
See~\cpageref{proof:coverability-sum-dim} in~\cref{app:coverability}.

\begin{remark}\label{rem:coverability-max-dim}
    As a corollary of~\cref{thm:coverability-sum-dim}, one can obtain double exponential length witnesses for coverability in terms of the SCC dimension and the number of states in the VASS.
    Recall that by~\cref{lem:sum-dim-le-max-dim-mult-num-states}, the product of $\maxdim$ and the number of states $n$ is an upper bound of $\sumdim$.
    Hence, if there is a run from $\Config{s}{x}$ that covers $\Config{t}{y}$, then there is a run from $\Config{s}{x}$ that covers $\Config{t}{y}$ of length at most $\poly(d,n,M,\infnorm{\vec{y}})^{f(n\cdot\maxdim)}$.
    Unfortunately, this means that even when the SCC dimension is fixed, the length of the shortest runs witnessing coverability can depend doubly-exponentially on the number of states $n$.
\end{remark}

\subsection{Improving Coverability Witness Length Bounds}

In this section we continue to study the coverability problem in VASS parameterised by geometric dimension $\sumdim$. 
Our main result here is an improved version of Theorem~\ref{thm:coverability-sum-dim} with the function $h(x) = 2^{x+1}-1$ improving upon function $f(x) = (x+1)^x$ from Theorem~\ref{thm:coverability-sum-dim}.

\begin{restatable}{theorem}{improvedCoverabilitySumDim}\label{thm:improved-coverability-sum-dim}
    Let $h: \NN \to \NN$, $h(x) \coloneqq 2^{x+1}-1$.
    If there is a run in $V$ from $\Config{s}{x}$ that covers $\Config{t}{y}$, then there exists a run in $V$ from $\Config{s}{x}$ that covers $\Config{t}{y}$ of length at most $\poly(d,n,M,\infnorm{\vec{y}})^{h(\sumdim)}$.
\end{restatable}

Similar to how the proof of Theorem~\ref{thm:coverability-sum-dim} was based on the ideas introduced by Rackoff in~\cite{Rackoff78}, the proof of Theorem~\ref{thm:improved-coverability-sum-dim} is based on the ideas from a recent paper by Kunnemann et.\ al.~\cite{KunnemannMSSW23}, where it was shown that the length of a covering run can be bounded by $\poly(d, n, M, \infnorm{\vec{y}})^{2^{\Oo(d)}}$. Similarly as in~\cite{Rackoff78} the construction from~\cite{KunnemannMSSW23} uses an induction on the dimension $d$. It also uses the fact that if a counter becomes ``sufficiently large'' then there is no need to worry that this counter will drop below zero and we can simply ignore this counter. The main contribution of~\cite{KunnemannMSSW23} is an observation that if a few counters become sufficiently large together we can ignore all of them at once. Therefore, the threshold above which we ignore them is lower.
In other words: in~\cite{Rackoff78} a configuration is treated as large if one of its counters exceeds some threshold $C$, while in~\cite{KunnemannMSSW23} a configuration is treated as large if either one of its counters exceeds an appropriate threshold $C_{d-1}$, or two of its counters exceed some smaller threshold $C_{d-2}$, or three of its counters exceed some even smaller threshold $C_{d-3}$, etc., or all of its counters exceed a quite small threshold $C_0$. Towards generalising this result we call vectors which do not satisfy this assumption $\vec{C} = (C_0, \ldots, C_{d-1})$-bounded and generalise this notion to $\vec{C}$-thin.

\begin{definition}[Bounded and thin configurations]
\label{def:thin-configurations}
Let $\vec{C} = (C_0, C_1, \ldots, C_{g-1}) \in \N^g$ be a vector of constants for some $C_0 \leq C_1 \leq \ldots \leq C_{g-1}$.
A vector $\vec{v} \in \N^g$ is \emph{$\vec{C}$-bounded} if for all $k \in [g]$ the number of coordinates $i \in [g]$ such that $\vec{v}[i] \geq C_{g-k}$ is strictly less than $k$.
A vector $\vec{v} \in \N^d$ is \emph{$\vec{C}$-thin}
with respect to $g$-dimensional vector space $U \subseteq \Q^d$ if there exists a clean basis $B$ of $U$ with set of distinguished coordinates $K$ such that $\proj{K}{\vec{v}}$ is $\vec{C}$-bounded.
Moreover, we say that a configuration $\Config{q}{v}$ of a VASS $V$ is \emph{$\vec{C}$-thin} if $\vec{v}$ is $\vec{C}$-thin with respect to $\cyclespace{V}$.
If a vector $\vec{v}$ or a configuration $\Config{q}{v}$ is not $\vec{C}$-thin then we call it \emph{$\vec{C}$-thick}.
\end{definition}

Similarly as before we prove that there are not too many $\vec{C}$-thin configurations reachable from a given configuration, we state it in Claim~\ref{clm:thin-configurations}.
We say that a vector $\vec{v} \in \N^k$ is \emph{sorted} if for each $i,j \in [k]$, $i < j$ we have $\vec{v}[i] \leq \vec{v}[j]$.

\begin{restatable}{claim}{thinConfigurations}\label{clm:thin-configurations}
For every configuration $\Config{p}{u}$ and every sorted $\vec{C} = (C_0, C_1, \ldots, C_{g-1}) \in \N^\sumdim$ there are at most $n \cdot d^\sumdim \cdot \prod_{i=0}^{\sumdim-1} C_i$ many $\vec{C}$-thin configurations reachable from $\Config{p}{u}$.
\end{restatable}

\begin{proof}
The proof is similar to the proof of Claim~\ref{clm:small-configurations}.
We argue as before that there are at most $d^\sumdim$ clean basis.
For a fixed clean basis and one of $n$ states fixed each configuration corresponds to a $\vec{C}$-bounded vector.
The number of such vectors is $\prod_{i=1}^\sumdim C_i$, therefore the number of $\vec{C}$-thin configurations
is bounded by $n \cdot d^\sumdim \cdot \prod_{i=1}^\sumdim C_i$, as required.
\end{proof}

In analogue of Kunnemann et.\ al., we would like to argue that if a configuration is $\vec{C}$-large, then there is a number of counters that are sufficiently large to be ``ignored''.
However, as in the proof of Theorem~\ref{thm:coverability-sum-dim} we need to prove something stronger. We show that if a configuration is $\vec{C}$-thick, then there is a collection of counters which are all sufficiently large and by ``ignoring'' them, the geometric dimension of the VASS decreases appropriately.

\begin{claim}\label{clm:c-thick}
    Let $U \sset \QQ^d$ be a $g$-dimensional vector space, let $\vec{v} \in \N^d$, and let $\vec{C} = (C_0, \ldots, C_{g-1}) \in \NN^g$ be sorted.
    If $\vec{v}$ is $\vec{C}$-thick with respect to $U$, then there exists a clean basis $B = \set{\vec{b}_1, \ldots, \vec{b}_g}$ of $U$, a number $k \in [g]$ and a set $S = \{i_1, \ldots, i_k\} \subseteq [g]$ of $k$ indices such that for every $i \in S$ and every $j \in [d]$, $\vec{b}_i[j] \neq 0$ implies that $\vec{v}[j] \geq C_{g-k}$. 
\end{claim}
\begin{proof}
    The proof of Claim~\ref{clm:c-thick} is similar to the proof of Claim~\ref{clm:c-large}. The difference is as follows. In Claim~\ref{clm:c-large} we need to find one vector $\vec{b}_i$ in some clean basis $B$ such that all the coordinates affected by $\vec{b}_i$ (namely $j \in [d]$ such that $\vec{b}_i[j] \neq 0$) have at least value $C$ in our vector $\vec{v}$. Here our task is more challenging. Besides the choice of $B$ we also have a choice of the number $k \in [g]$ indicating how many vectors from $B$ we plan to ignore. For a chosen $k \in [g]$ we need to find vectors $\vec{b}_{i_1}, \ldots, \vec{b}_{i_k}$ from $B$ such that all the coordinates affected by them have a value at least $C_{g-k}$.
    
    Even though the task seems more complicated the proof is essentially proceeding the same as the proof of Claim~\ref{clm:c-large}. We mainly present here the parts which differ.
    
    Assume, for sake of contradiction, that $\vec{v}$ is $\vec{C}$-thick and there does not exist a clean basis $B$, a number $k \in [g]$
    and a set of indices $S = \{i_1, \ldots, i_k\}$ such that
    for every index $i \in S$ and for every $j \in [d]$, $\vec{b}_i[j] \neq 0$ implies that $\vec{v}[j] \geq C_{g-k}$.
    Consider a clean basis $B = \set{\vec{b}_1, \ldots, \vec{b}_g}$ with a distinguished set of coordinates $K = \set{k_1, \ldots, k_g}$ of minimal $N \coloneqq \onenorm{\proj{K}{\vec{v}}}$.
    
    Now, our goal is to construct a new clean basis $B'$ with a distinguished set of coordinates $K'$ such that $\onenorm{\proj{K'}{\vec{v}}}$ is less than $N$.
    Since $\vec{v}$ is $\vec{C}$-thick we know that $\proj{K}{\vec{v}}$ is not $\vec{C}$-bounded. That means that there is a number $k \in [g]$ such that the number of coordinates $i \in [g]$ such that $\vec{v}[i] \geq C_{g-k}$ is at least $k$. Let $S = \{i_1, \ldots, i_k\}$ be a set of some $k$ coordinates such that $\vec{v}[i_\ell] \geq C_{g-k}$ for all $\ell \in [k]$. Now, using our assumption we get that there is some index $i \in S$ and some coordinate $j \in [d]$ such that $\vec{b}_i[j] \neq 0$, but $\vec{v}[j] < C_{g-k}$. That means that substituting the distinguished coordinate $k_i$ by the coordinate $j$ will decrease the sum of values $\vec{v}[i]$ for distinguished coordinates $i$, exactly as in the proof of Claim~\ref{clm:c-large}.
    
    The rest of the proof is literally the same as for Claim~\ref{clm:c-large}. We recall here the definition of the new basis $B'$.
    We set $B' = \set{\vec{b}'_1, \ldots, \vec{b}'_g}$ by
    \begin{equation*}
        \vec{b}'_i \coloneqq \frac{\vec{b}_i}{\vec{b}_i[j]} 
        \text{\; and \;}
        \vec{b}'_k \coloneqq \vec{b}_k - \vec{b}_k[j]\cdot\vec{b}'_i, \text{ for all } k \neq i.
    \end{equation*}
    We identify the set of distinguished coordinates $K' = \set{k_1, \ldots, k_{i-1}, j, k_{i+1}, \ldots, k_g}$.
    
    Verification that $B'$ is a clean basis is performed exactly as in the proof of Claim~\ref{clm:c-large}. Finally, we also verify that $\onenorm{\proj{K'}{\vec{v}}} < N$ exactly in the same way as before. The intuition is that $\vec{v}[k_i]$ was replaced by a smaller value $\vec{v}[j]$, thus $\onenorm{\proj{K'}{\vec{v}}}$ decreased. This contradicts minimality of $N$ and finishes the proof.
\end{proof}

With the understanding that in a $\vec{C}$-thick configuration one can ignore some coordinates to decrease the geometric dimension we are ready to prove~\cref{thm:improved-coverability-sum-dim}.
We first define the bounds on the length of runs that we would like to establish.
\begin{equation*}
    K_i \coloneqq
    \begin{cases}
        n - 1 & \text{ for } i = 0 \\
        n \cdot d^i \cdot \prod_{j=0}^{i-1} (\infnorm{\vec{y}} + M \cdot K_j) + K_{i-1} & \text{ for } i \geq 1
    \end{cases}
\end{equation*} 
Our goal is to prove that coverability is witnessed by runs of length at most $K_{\sumdim}$.

\begin{lemma}\label{lem:henry}
    If there is run in $V$ from $\Config{s}{x}$ that covers $\Config{t}{y}$, then there is a run from $\Config{s}{x}$ that covers $\Config{t}{y}$ of length at most $K_{\sumdim}$.
\end{lemma}

\begin{proof}
    The proof proceeds similarly to the proof of Lemma~\ref{lem:rackoff}, so we focus on the differences while only sketching the identical parts.
    
    The proof is by induction on the geometric dimension $\sumdim$.
    The case of $\sumdim = 0$ is easy and the same as for Lemma~\ref{lem:rackoff}.
    For $\sumdim \geq 1$ we assume that Lemma~\ref{lem:henry} holds for VASS with geometric dimension less than $g$.
    Now, assume that there is a run in $V$ from $\Config{s}{x}$ that covers $\Config{t}{y}$.
    Let $\Config{s}{x} \xrightarrow{\pi} \Config{t}{y'}$, for some $\vec{y}' \geq \vec{y}$, be the shortest run from $\Config{s}{x}$ that covers $\Config{t}{y}$.
    We will split this run at the first configuration that is $\vec{C}$-thick for $\vec{C} = (C_0, C_1, \ldots, C_{g-1})$,
    where $C_i = M \cdot K_i + \infnorm{\vec{y}}$.
    Precisely, consider the decomposition: $\Config{s}{x} \xrightarrow{\pismall} \Config{q}{m} \xrightarrow{\pitail} \Config{t}{y'}$, where $\Config{q}{m}$ is the first $\vec{C}$-thick configuration.
    
    Since $\Config{s}{x} \xrightarrow{\pi} \Config{t}{y'}$ has minimal length, we know that no configuration is repeated.
    As all the configurations before $\Config{q}{m}$ are $\vec{C}$-thin
    we know by Claim~\ref{clm:thin-configurations} that
    the length of $\pismall$ is at most $n \cdot d^\sumdim \cdot \prod_{i=0}^{\sumdim-1} C_i$.
    
    Now let us focus on bounding the length of $\pitail$. Since $\Config{q}{m}$ is $\vec{C}$-thick by Claim~\ref{clm:c-thick}
    we know that there is a clean basis $B = \set{\vec{b}_1, \ldots, \vec{b}_g}$ of $U$, a number $k \in [g]$ and a set $S = \{i_1, \ldots, i_k\} \subseteq [g]$ of $k$ indices such that for every $i \in S$ and every $j \in [d]$, $\vec{b}_i[j] \neq 0$ implies that $\vec{v}[j] \geq C_{g-k}$.
    In other words it means that all the coordinates influenced by basis vectors $\vec{b}_{i_1}, \ldots, \vec{b}_{i_k}$
    are of size at least $C_{g-k}$.
    
    We now sketch bounding $\pitail$ without going into details, since the details are the same as in Lemma~\ref{lem:rackoff}.
    We consider a VASS $V'$, which is exactly $V$ with coordinates in $J$ removed, where $J$ contains all the coordinates
    on which one of the $\vec{b}_{i_j}$ is non-zero (for $j \in [k]$). Let $J' = [d] \setminus J$ be the complement of $J$.
    Removing all the coordinates in $J$ implies that $\dim(\cyclespace{V'}) \leq \dim(\cyclespace{V})-k$ decreased by $k$, since all the basis vectors
    $\vec{b}_{i_j}$ for $j \in [k]$ are zero vectors after removing coordinates in $J$.
    Therefore $V'$ has geometric dimension at most $g-k$ and by induction assumption there is a covering run
    from $\Config{q}{\proj{J'}{m}}$ to $\Config{t}{\proj{J'}{y}}$ of length at most $K_{g-k}$. Now consider this run lifted back to $V$,
    with coordinates from $J$ added back. Since max-norm of $V$ is $M$ the run of length $K_{g-k}$ decreases counters in $J$
    by at most $M \cdot K_{g-k}$. Recall however that for each $i \in J$ we have $\vec{m}[i] \geq C_{g-k} = M \cdot K_i + \infnorm{\vec{y}}$.
    That means that after decreasing by $M \cdot K_{g-k}$ all the counters in $J$ still have value at least $\infnorm{\vec{y}}$.
    This is sufficient to cover $\Config{t}{y}$, so indeed the presented run covers $\Config{t}{y}$.
    Its length is at most the length of $\pismall$ plus $K_{g-k} \leq K_{g-1}$, which equals
    $n \cdot d^\sumdim \cdot \prod_{i=0}^{\sumdim-1} C_i + K_{g-1} = K_g$, as required.
\end{proof}

In order to prove~\cref{thm:improved-coverability-sum-dim}, we only need to argue that $K_{\sumdim}$ is
indeed bounded by $\poly(d,n,M,\infnorm{\vec{y}})^{h(\sumdim)}$, which we show in the following claim.

\begin{claim}\label{cl:h-bounded}
For any $g \in \N$ we have $K_g \leq \poly(d,n,M,\infnorm{\vec{y}})^{h(g)}$ for $h(x) = 2^{x+1}-1$.  
\end{claim}

\begin{proof}
    Recall that
    \begin{equation*}
       K_i \coloneqq
        \begin{cases}
          n - 1 & \text{ for } i = 0, \\
          n \cdot d^i \cdot \prod_{j=0}^{i-1} (\infnorm{\vec{y}} + M \cdot K_j) + K_{i-1} & \text{ for } i \geq 1.
       \end{cases}
    \end{equation*}
    Let $A = 4ndM (\infnorm{y}+1)$.
    We will prove that \hbox{$K_g \leq A^{h(g)}$} by induction on $g$.
    For $g = 0$, we have $(4ndM (\infnorm{y}+1))^{h(0)} = 4ndM (\infnorm{y}+1) \geq 4n > n-1 = K_0$.
    For $g \geq 1$, assume that $K_i \leq A^{h(i)}$ for all $i < g$.
    Now, by definition, we know that $K_g = n \cdot d^g \cdot \prod_{j=0}^{g-1} (\infnorm{\vec{y}} + M \cdot K_j) + K_{g-1}$.
    Therefore,
    \begin{align*}
        K_g 
        & = n \cdot d^g \cdot \prod_{j=0}^{g-1} (\infnorm{\vec{y}} + M \cdot K_j) + K_{g-1}
        \leq 2n \cdot d^g \cdot \prod_{j=0}^{g-1} (\infnorm{\vec{y}} + M \cdot K_j) \\
        & \leq \prod_{j=0}^{g-1} 2nd (\infnorm{\vec{y}} + M \cdot K_j)
        \leq \prod_{j=0}^{g-1} 4nd \cdot \infnorm{\vec{y}} \cdot M \cdot K_j \\
        & \leq \prod_{j=0}^{g-1} 4nd\cdot (\infnorm{\vec{y}}+1) \cdot M \cdot K_j
        = \prod_{j=0}^{g-1} A \cdot K_j
        \leq \prod_{j=0}^{g-1} A \cdot A^{2^{j+1}-1} \\
        & = \prod_{j=0}^{g-1} A^{2^{j+1}} = A^{2^{g+1}-1} = A^{h(g)},
    \end{align*}
    where the third inequality holds because $a + b \leq 2ab$ for all $a, b \in \N$ and the last inequality holds by induction assumption.
    We conclude as $K_g \leq (4ndM(\infnorm{\vec{y}} + 1))^{h(g)}$ is bounded by $\poly(d,n,M,\infnorm{\vec{y}})^{h(g)}$.
\end{proof}

\section{Simultaneous Unboundedness}
\label{sec:simultaneous-unboundedness}
Viewed as a decision problem, simultaneous unboundedness asks, from a given initial configuration $\Config{s}{x}$, whether it is possible to reach a given target state $q$ with counter values $\vec{v} \geq (G, \ldots, G)$, for a given target value $G \in \NN$. 
One can observe that if such a run does not exist, then there is a bound $H$ such that, for every run, there is a counter which does not exceed $H$. 
This observation is a necessary ingredient in the standard approach for solving reachability in VASS (i.e.\ this is used as a tool in the KLM decomposition)~\cite{Kosaraju82, LerouxS19}. 
Simultaneous unboundedness, as well as other variants of (un)boundedness have already been studied in VASS parameterised by the dimension $d$~\cite{Demri13,CzerwinskiHZ18}.
In this section, we study simultaneous unboundedness in VASS parameterised by the geometric dimension $\sumdim$. 

Throughout this section, unless otherwise specified, we fix our attention to a VASS $V = (Q, T)$.
We also fix the following parameters of $V$:
\begin{itemize}
        \item $d \in \NN$ is the dimension,
        \item $\sumdim \leq d$ is the geometric dimension, 
        \item $n \coloneqq \abs{Q}$ is the number of states,
        \item $M \coloneqq \max\set{\infnorm{\vec{x}} : (p, \vec{x},q) \in T}$ is the magnitude of the greatest effect of any transition.
\end{itemize}

We will prove that, for all \(G \in \N\), there are \(H, L \in \N\) of size \( \poly(d,n,G,M)^{f(\sumdim)}\) for some exponential function $f$ s.t. if there is a run in which every counter observes a value \(\geq H\) (not necessarily at the same time), then there is a run of length \(\leq L\) in which all counters are \(\geq G\) \emph{simultaneously}. 
The relationship between $G$, $H$, and $L$ is detailed in~\cref{thm:no-pump} and their precise values are provided ahead of~\cref{lem:no-pump-concrete}.

\begin{restatable}{theorem}{noPump}\label{thm:no-pump}
    Let $f: \N \to \N$, $f(x) \coloneqq (x+1)^{x+1}$. 
    For every $G \in \N$, there exist $H, L \in \NN$ such that $H, L \leq \poly(d,n,G,M)^{f(\sumdim)}$ and, if there exists a run from $\Config{s}{x}$ to $\Config{q}{u}$ in which, for every $i \in [d]$, there is a configuration $p_i(\vec{h}_i)$ with $\vec{h}_i[i] \geq H$, then there exists a run from $\Config{s}{x}$ to $\Config{q}{v}$ of length at most $L$, such that $\vec{v} \geq (G, \ldots, G)$.
\end{restatable}

The proof of~\cref{thm:no-pump} is very similar to the proof of~\cref{thm:coverability-sum-dim} which is based on Rackoff's proof that coverability is witnessed by a double exponential length run~\cite{Rackoff78}.
We will sketch the ideas behind the proof of~\cref{thm:no-pump} first for the scenario in which the VASS is parameterised by its dimension $d$ (not its geometric dimension).
Given the assumption that there is a run in which every counter exceeds the threshold $H$ (at some point), we can  ``ignore'' the first counter that exceeds this threshold.
This idea is that, if in a $(d-1)$-VASS, we can prove that there is a run that witnesses simultaneous unboundedness of length $L$, and we know that $H \geq M \cdot L + G$, then we can ``safely ignore'' the counter that exceeded the threshold $H$.
Accordingly, we proceed by induction on the dimension with appropriately chosen values for $H$ and $L$. 
In our case, we are considering VASS parameterised by the geometric dimension $g$.
Just like the proof of~\cref{thm:coverability-sum-dim}, we instead wish to ignore multiple counters at the same time in order to decrease the geometric dimension by 1 for the inductive step.
For this, we consider two levels of ``large'': a \emph{large} upper bound $C$ and \emph{very large} threshold $H$.
We use the idea that, from a geometrically $C$-large configuration, we can ``safely ignore'' a collection of counters that will decrease the geometric dimension.
However, \emph{after} the first $C$-large configuration is observed, we have to prove that we still fulfill the induction hypothesis that there is a run which observes a large value on every counter. In order to still observe configurations that exceed the threshold, we need to set the threshold $H$ to be (significantly) larger than $C$.
Next, we will precisely define the threshold and upper bound values.
Fix $G \in \N$.
Let us inductively define three sequences of numbers.

\begin{equation*}
    C_i = 
    \begin{cases}
        0                   & i = 0, \\
        M\cdot L_{i-1} + G  & i \geq 1.
    \end{cases}
    \hspace{10mm}
    H_i = 
    \begin{cases}
        n\cdot (d+1) \cdot  M + G & i = 0, \\
        n \cdot M \cdot (d\cdot C_i)^i + H_{i-1}   & i \geq 1.
    \end{cases}
\end{equation*}

\begin{center}
    \begin{equation*}
        L_i = 
        \begin{cases}
            n \cdot (d+1)  & i = 0, \\
            n \cdot (d\cdot C_i)^i  + L_{i-1}   & i \geq 1.
        \end{cases}
    \end{equation*}
\end{center}

\begin{lemma}\label{lem:no-pump-concrete}
    If there exists a run
    from $\Config{s}{x}$ to $\Config{q}{u}$ in which, for every $i \in [d]$, there is a configuration $p_i(\vec{h}_i)$ with $\vec{h}_i[i] \geq H_\sumdim$, then there exists a run from $\Config{s}{x}$ to $\Config{q}{v}$ of length at most $L_\sumdim$, such that $\vec{v} \geq (G, \ldots, G)$.
\end{lemma}

\begin{proof}
    In this proof, we shall refer to the configurations $p_i(\vec{h}_i)$ as \emph{high configurations}.
    
    This proof goes by induction on $\sumdim$. 
    For $\sumdim = 0$, consider a run from $\Config{s}{x}$ to $\Config{q}{u}$ and consider every high configuration $p_i(\vec{h}_i)$.
    We shall mark the starting configuration $\Config{s}{x}$, every high configuration $p_i(\vec{h}_i)$, and the final configuration $\Config{q}{u}$.
    Since $\sumdim = 0$, all cycles have zero effect.
    This means that, between every consecutive marked configuration, we can remove simple cycles.
    We remark that we cannot simply remove simple cycles from the entire run, as it may be the case that a high configuration is observed by using a cycle.
    Indeed, a particular counter may exceed the threshold $H_0$ by making use of a cycle that may change its value throughout the course of the cycle (even though the cycle overall has zero effect).
    Now, we know that there is a run from $\Config{s}{x}$ to $\Config{q}{u}$ which observes every high configuration $p_i(\vec{h}_i)$, and whose length is bounded by $n(d+1) = L_0$.
    Now, we will argue that $\vec{u} \geq (G, \ldots, G)$.
    Consider the $i$-th counter: it is at least $H_0 = n\cdot (d+1) \cdot M + G$ at $p_i(\vec{h}_i)$.
    Notice that the run from $p_i(\vec{h}_i)$ to $\Config{q}{u}$ has length at most $(d+1)\cdot n$.
    This means that at the end of the run, the $i$-th counter will be at least $n \cdot (d+1) \cdot M + G - (d+1)\cdot n \cdot M \geq G$ (as required).
    
    For the induction step, assume that Lemma~\ref{lem:no-pump-concrete} is true for all geometric dimensions less than $\sumdim$. 
    Now suppose that there is a run $\Config{s}{x} \xrightarrow{\rho} \Config{q}{u}$ in which, for every $i \in [d]$, there is a configuration $p_i(\vec{h}_i)$ with $\vec{h}_i[i] \geq H_\sumdim$. 
    Recall~\cref{def:small-configurations} and split $\rho$ at the first $C_\sumdim$-large configuration $p(\vec{m})$ to obtain
    $\Config{s}{x} \xrightarrow{\rhosmall} \Config{p}{m} \xrightarrow{\rhotail} \Config{q}{u}$.
    In fact, let $\Config{s}{x} \xrightarrow{\rho} \Config{q}{u}$ be the shortest such run. 
    
    By the minimality of the length of $\Config{s}{x} \xrightarrow{\rho} \Config{q}{u}$, we know that no configuration is repeated.  
    Thus, by \Cref{clm:small-configurations}, we know that $|\rho_{\text{small}}| \leq n \cdot (d \cdot C_\sumdim)^\sumdim$. 
    
    Moreover, by~\cref{clm:c-large}, there is a clean basis $B = \{\vec{b}_1, \ldots, \vec{b}_\sumdim\}$ and an index $i \in [\sumdim]$ such that, for every $j \in [d]$, if $\vec{b}_i[j] \neq 0$ then $\vec{m}[j] \geq C_\sumdim$. 
    In the same way as in the proof of~\cref{lem:rackoff}, we will project away all counters $J = \set{j \in [d] : \vec{b}_i[j] \neq 0}$ to obtain a VASS $V'$ of dimension $d - \abs{J}$ but, crucially, of geometric dimension $\leq \sumdim - 1$.
    
    Additionally, for every $i \in [d]$, there is a high configuration $p_i(\vec{h}_i)$ observed during $\Config{s}{x} \xrightarrow{\rho} \Config{q}{u}$ such that $\vec{h}_i[i] \geq H_\sumdim$. 
    Our goal is to replace $\rhotail$ with a sufficiently short run using the inductive assumption.
    Roughly speaking, after projecting away counters in $J$, we wish to identify a run in $V'$ starting from $p(\proj{[d]\setminus J}{\vec{m}})$ which ends with all counter values at least $G$.
    To do this, we need to identify configurations in $\Config{p}{m} \xrightarrow{\rhotail} \Config{q}{u}$ which are high. 
    Specifically, for every $i \in [d] \setminus J$, we need to identify a configuration $p_i(\vec{h}_i)$ that occurs in $\Config{p}{m} \xrightarrow{\rhotail} \Config{q}{u}$ that satisfies $\vec{h}_i[i] \geq H_{\sumdim-1}$.
    We remark that the original high configuration $p_i(\vec{h}_i)$ may be observed during $\rhosmall$ or $\rhotail$; clearly if the high configuration $p_i(\vec{h}_i)$, for $i \in [d] \setminus J$, is observed during $\rhotail$, then it satisfies $\vec{h}_i[i] \geq L_{\sumdim} > L_{\sumdim-1}$.
    For the high configurations $p_i(\vec{h}_i)$ that are observed during $\rhosmall$, we will prove that there is a configuration observed during $\rhotail$ that is sufficiently high on the same counter.
    In fact, we will prove that this configuration is $\Config{p}{m}$.
    Namely, for every $i \in [d] \setminus J$, for which $p_i(\vec{h}_i)$ occurs in $\Config{s}{x} \xrightarrow{\rhosmall} \Config{p}{m}$, it is true that $\vec{m}[i] \geq H_{\sumdim-1}$.
    This is because $H_\sumdim = n \cdot (d \cdot C_\sumdim)^\sumdim \cdot M + H_{g-1} \geq \abs{\rhosmall} \cdot M + H_{g-1}$.
    
    Therefore, we can apply the induction assumption on $V'$ starting from $p(\proj{[d]\setminus J}{\vec{m}})$ to find a run $p(\proj{[d]\setminus J}{\vec{m}}) \xrightarrow{\pi} \Config{q}{u'}$, of length at most $L_{\sumdim-1}$, such that $\vec{u}' \geq (G, \ldots, G)$. 
    Since all the ignored counters have value at least $C_\sumdim = M \cdot L_{\sumdim-1} + G$ in the configuration $p(\vec{m})$, we deduce that the run $p(\vec{m}) \xrightarrow{\pi'} q(\vec{w})$ in $V$ obtained by following the same path as $\pi$ has length at most $L_{\sumdim-1}$ and $\vec{w} \geq (G, \ldots, G)$.
    Finally, the concatenation of $\rhosmall$ and $\pi'$ has length at most $n \cdot (d \cdot C_\sumdim)^\sumdim  + L_{\sumdim-1} = L_\sumdim$.
\end{proof}

Now, to conclude this section by proving~\Cref{thm:no-pump}, we only need to argue that $L_g$ and $H_g$ are indeed bounded by $\poly(d,n,M,G)^{f(\sumdim)}$. See~\Cref{app:simultaneous}.

\section{Boundedness}
\label{sec:boundedness}
In this section we study the boundedness problem in VASS parameterized by geometric dimension, which asks whether infinitely many configurations can be reached from a fixed source. We first recall that unboundedness is witnessed by any run that contains a configuration strictly covering a previous configuration.
Consider a run of the form $c_0 \xrightarrow{u_1} c_1 \xrightarrow{u_2} c_2 \cdots \xrightarrow{u_m} c_m$. We say it is an \emph{unboundedness witness for $c_0$} if there is an index $i < m$ such that $c_i < c_m$.
In this case, observe that the cycle from $c_i$ to $c_m$ can be fired any number of times from $c_m$ to reach infinitely many configurations. On the other hand, assume the reachable set from $c_0$ is indeed infinite. The configurations reachable from $c_0$ naturally form a finite-branching tree with root $c_0$. By K\"onig's Lemma there must be an infinite branch in this tree. Using the fact that $\leq$ is a well-quasi-order over configurations, we are able to extract an unboundedness witness from this branch.

\begin{lemma}[{\cite[Lemma 4.2]{Rackoff78}}]
    Let $V$ be a VASS and $c$ be a configuration of $V$. Then $\{c' \mid c \xrightarrow{*} c'\}$ is infinite if and only if there is an unboundedness witness for $c$.
\end{lemma}

We remark that unboundedness witnesses were called \emph{self-covering paths} in \cite{Rackoff78}. Now the boundedness problem is reduced to finding unboundedness witnesses. Our result, resembling that of Rackoff \cite{Rackoff78}, gives length bounds on shortest unboundedness witnesses in terms of the geometric dimension.
Similarly as in Section~\ref{sec:coverability} we fix the following parameters of $V$:
\begin{itemize}
        \item $d \in \NN$ is the dimension,
        \item $\sumdim \leq d$ is the geometric dimension, 
        \item $n \coloneqq \abs{Q}$ is the number of states,
        \item $M \coloneqq \max\set{\infnorm{\vec{x}} : (p, \vec{x},q) \in T}$ is the magnitude of the greatest effect of any transition.
\end{itemize}

\begin{restatable}{theorem}{boundedness}
    \label{thm:boundedness}
    Let $f: \N \to \N$, $f(x) \coloneqq (4x+2)^{2x+1}$.
    If there is an unboundedness witness in VASS $V$
    for some configuration $s$ then there is also one of length bounded by $\poly(d, n, M)^{f(\sumdim)}$.
\end{restatable}

\begin{remark}
Observe that in \cref{thm:boundedness} if we fix the geometric dimension $\sumdim$ then the corresponding bound is polynomial w.r.t.\ $d, n,$ and $M$. Moreover, the bounds do not depend on the configuration $s$.
\end{remark}

\begin{remark}\label{rem:boundedness-max}
    Similar to \cref{rem:coverability-max-dim}, using \cref{lem:sum-dim-le-max-dim-mult-num-states} we can obtain a length bound on unboundedness witnesses in terms of the SCC dimension $\maxdim$ and the number of states, namely $\poly(d, n, M)^{f(n \cdot \maxdim)}$.
\end{remark}

\subsection{Length Bound of Unboundedness Witnesses}

We will actually prove \cref{thm:boundedness} for a generalized model called \emph{VASS extended with integer counters}, which allow some counters to take negative values in a run. 

Formally, a $d$-VASS extended with integer counters is given by a triple $V_{\Z} = (V, I_\N, I_\Z)$ where $V = (Q, T)$ is a $d$-VASS and $I_\N, I_\Z$ is a partition of $[d]$. 
We remark that the underlying VASS $V$ will be the one under consideration of this section. And we use its parameters $g, n, M$ as fixed at the beginning of this section.
The set of legal configurations of $V_\Z$ is given by $\{q(\vec{x}) \in Q \times \Z^d \mid \vec{x}[i] \in \N \text{ for all } i \in I_\N\}$. That is, counters indexed by $I_\Z$ are allowed to take negative values while other counters must be kept nonnegative. 
Unboundedness witnesses are defined in the same way as for standard VASS. So an unboundedness witness for configuration $c$ is a run $\pi = \pi_1 \pi_2$ starting from $c$ such that $\pi_2$ is a cycle with positive effect $\eff(\pi_2) > \vec{0}$. However, as opposed to prior sections, VASS extended with integers counters have two natural notions of geometric dimension. We define the geometric dimension of $V_\Z$ to be the geometric dimension $\sumdim$ of $V$, i.e.\ using all counters. We also define the \emph{restricted geometric dimension} of $V_\Z$, denoted $\ressumdim(V_\Z)$, to be the geometric dimension of the VASS obtained from $V$ by removing all counters in $I_\Z$. That is, 
\begin{equation}
    \ressumdim(V_\Z) = \dim( \proj{I_\N}{\cyclespace{V}} ).
\end{equation}

We will prove \cref{thm:boundedness} for VASS extended with integer counters. As first step, we define the bounds on the length of unboundedness witnesses that we would like to establish.
\begin{align*}
    L_i = \begin{cases}
        D^{g+1} & \text{for } i = 0 \\
        (D \cdot (d \cdot M \cdot L_{i-1})^{4i})^{g+1} + L_{i-1} & \text{for } i \ge 1
    \end{cases}
\end{align*}
where $D = (5d^2n^2M)^2$. We are going to prove the following lemma.

\begin{lemma}
    \label{lem:boundedness-vass-z}
    If there is an unboundedness witness in $V_\Z = (V, I_\N, I_\Z)$ for some configuration $s$, then there is also one of length bounded by $L_{k}$ where $k = \ressumdim(V_\Z) \leq \sumdim$.
\end{lemma}

For a standard VASS $V$ we have $\ressumdim(V) = \sumdim$. Therefore, \cref{lem:boundedness-vass-z} bounds the length of shortest unboundedness witnesses by $L_{\sumdim}$. In \cref{app:boundedness} we verify that this implies \cref{thm:boundedness}, i.e.\ we do the calculation to show that \(L_{\sumdim}\) is indeed at most double exponential.

The rest of this section is devoted to proving \cref{lem:boundedness-vass-z}. Let $V_\Z = (V, I_\N, I_\Z)$ be the considered $d$-VASS extended with integer counters, and $s$ be a configuration of $V$ such that there is an unboundedness witness $\pi$ for $s$. Suppose $V = (Q, T)$.
Moreover, we may assume w.l.o.g.\ that $V$ contains only those transitions in $\pi$. Under this assumption, the structure of $V$ is simply a series of SCCs, and \cref{clm:path-difference} can be applied.
In the following we show Lemma~\ref{lem:boundedness-vass-z} by induction on the restricted geometric dimension $k = \ressumdim(V)$. 

\paragraph*{The Base Case}
\label{sec:boundedness-base-case}

First we consider the case that $k = 0$, which means that every cycle $\theta$ in $V$ satisfies $\proj{I_\N}{\eff(\theta)} = \vec{0}$. Observe however that transitions may still have effects on \(\N\)-counters, hence it is not obvious that they can be ignored. The following claim crucially uses our above assumption that \(V\) contains only those transitions on the run \(\pi\).

\begin{restatable}{claim}{NCountersIngored}
    \label{cl:N-counters-ignored-res-dim-eq-zero}
    Any path $\rho$ in $V$ can be lifted to a run $s \xrightarrow{\rho} c$ for some configuration \(c\). Moreover, for any reachable configuration $q(\vec{y})$ and any configuration $q(\vec{y}_\pi)$ on the unboundedness witness $\pi$ (with the same state), we have $\proj{I_\N}{\vec{y}_\pi} = \proj{I_\N}{\vec{y}}$.
\end{restatable}

By Claim~\ref{cl:N-counters-ignored-res-dim-eq-zero}, in order to find an unboundedness witness it is enough to find paths with certain effects, since every path can be lifted to a run. Therefore in the following we assume $I_\N = \emptyset$, that is, $V_\Z$ has only integer counters. We decompose the unboundedness witness $\pi$ as $\pi = \pi_1\pi_2$ such that $\pi_2$ is a cycle with $\eff(\pi_2) > \vec{0}$. Since all counters are integer counters, we can replace $\pi_1$ by a simple path connecting its source and target states, so that $|\pi_1| \leq n$. We now express $\pi_2$ as a solution of a system of linear inequalities, and rely on~\cref{clm:linear_inequalities} to obtain a bound of its length. Let $Q' \subseteq Q$ and $T' \subseteq T$ be the subsets of states and transitions that appear in $\pi_2$. Simple cycles that use only transitions from $T'$ will be called $T'$-cycles. We first shrink $\pi_2$ into a short run $\sigma$ as follows: whenever there is a simple cycle $\theta$ as an infix of $\pi_2$, such that every state visited by $\theta$ is also visited by the prefix of $\pi_2$ before $\theta$, we remove $\theta$ from $\pi_2$. Notice that $\sigma$ also visits every state in $Q'$. Also observe that $|\sigma| \le n^2$ as $\sigma $ visits every state at most $n$ times. Let $\mathcal{S}$ be the set of simple cycles removed from $\pi_2$. For each cycle $\theta \in \mathcal{S}$ we have $\NORM(\eff(\theta)) \leq nM$ as it is a simple cycle. Observe that $\eff(\pi_2) - \eff(\sigma) \in \intcone(\eff(\mathcal{S}))$. Therefore, because of \Cref{thm:caratheodory-integer} we can choose $\mathcal{C} \subseteq \mathcal{S}$ such that $|\mathcal{C}| \leq 2d\log(4dnM) \leq 4d^2nM$ and $\eff(\pi_2) - \eff(\sigma) \in \intcone(\eff(\mathcal{C}))$. We define a system $U$ of $d$ linear inequalities, with unknowns $x_\theta$ correspond to $T'$-cycles $\theta$ from $\mathcal{C}$. Concretely, for $i \in [1,d]$ we have inequality:
\begin{align}
    \sum_{\theta \in \mathcal{C}}x_\theta \cdot \eff(\theta)[i] + \eff(\sigma)[i] \geq [\eff(\pi_2)[i] > 0]
    \label{eq:ilp-unboundedness}
\end{align}
where $[]$ is the Iverson bracket.

Any solution of $U$ induces a run $\pi_2'$ obtained by attaching to $\sigma$ each cycle $\theta \in \mathcal{C}$ repeated for $x_\theta$ times. This is possible as $\sigma$ visits every state in $Q'$. Also, \cref{eq:ilp-unboundedness} ensures that $\eff(\pi_2') > \vec{0}$. Thus $\pi' := \pi_1\pi_2'$ is also an unboundedness witness for $s$.

Clearly $U$ has a nonnegative integer solution, namely the one obtained from $\pi_2$ by setting $x_\theta$'s to be the coefficients witnessing $\eff(\pi_2) - \eff(\sigma) \in \intcone(\eff(\mathcal{C}))$. Suppose $U$ is expressed in the form $A\vec{x} \ge \vec{b}$. Then columns of $A$ are effects of simple cycles. We conclude that $\rank(A) \le g$. Also, each entry in $A$ and $\vec{b}$ is bounded in absolute value by $Mn^2$. Hence, by \cref{clm:linear_inequalities}, we can bound the max-norm of minimal solutions of $U$ by $(g+1)(g\cdot Mn^2)^{g} \leq (2gMn^2)^{g+1} \leq (2dMn^2)^{g+1}$. The length of $\pi_2'$ yielded by any minimal solution is then bounded by $(4d^2nM) \cdot n \cdot (2dMn^2)^{g+1} + n^2 \le (4d^2n^2M)^{2g+2}$.
Therefore, the length of the new unboundedness witness $\pi_1\pi_2'$ is bounded by $(4d^2n^2M)^{2g+2} + n \leq (5d^2n^2M)^{2g+2} = L_{0}$. This completes the proof of the base case.

\paragraph*{The Induction Step}

Now we consider the case that $k = \ressumdim(V_\Z) > 0$. Let $V_\N$ be the $d'$-VASS obtained from $V_\Z$ by removing all counters in $I_\Z$, so $d' = |I_\N|$. Recall the definition of $C$-small/large configurations for standard VASS (\cref{def:small-configurations}). Here we say a configuration $q(\vec{x})$ in $V_\Z$ is $C$-small/large if the projected configuration $q(\proj{I_\N}{\vec{x}})$ is $C$-small/large in $V_\N$. Take $C = L_{k-1} \cdot M$. We consider two cases depending on whether the unboundedness witness $\pi$ for $s$ contains a $C$-large configuration.

\textbf{Case 1.} Assume $\pi$ contains a $C$-large configuration. Let $\Config{q}{m}$ be the first $C$-large configuration on $\pi$, which splits $\pi$ into $s \trans{\pismall} \Config{q}{m} \trans{\pitail} t$.
By definition of unboundedness witness, $\pi$ can also be factored into $\pi = \pi_1\pi_2$ with $\pi_2$ being a cycle satisfying $\eff(\pi_2) > \vec{0}$. Observe that $\pi' := \pitail\pi_2$ is an unboundedness witness for $\Config{q}{m}$. The strategy is to use the induction hypothesis to replace $\pi'$ by a length-bounded run. 

Since $\Config{q}{m}$ is $C$-large, by \cref{clm:c-large} there exists a clean basis $B = \{\vec{b}_1, \ldots, \vec{b}_k\}$ of the vector space $\cyclespace{V_\N}$ and an index $i \in [k]$, such that for all $j \in I_\N$, $\vec{m}[j] \ge C$ whenever $\vec{b}_i[j] \ne 0$. Let $I := \{j \in I_\N \mid \vec{b}_i[j] \ne 0\}$. We turn every counter in $I$ into an integer counter. That is, we consider the VASS $V_\Z' := (V, I_\N \setminus I, I_\Z \cup I)$. It is easy to see that $\ressumdim(V_\Z') < \ressumdim(V_\Z) = k$, and that $\pi'$ is still an unboundedness witness for $\Config{q}{m}$ in $V_\Z'$. By induction hypothesis, there is also an unboundedness witness $\rho$ for $\Config{q}{m}$ in $V_{\Z}'$ with $|\rho| \le L_{k-1}$. We claim that $\rho$ induces a legal run in $V_\Z$ from $\Config{q}{m}$. For this we only need to consider counters in $I$. Notice that for any $j \in I$, $\vec{m}[j] \ge C \ge |\rho| \cdot M$, while for any prefix $\rho$ of $\rho$ we have $\eff(\rho')[j] \le |\rho| \cdot M$. So values of counters in $I$ will not drop below zero along $\rho$. This shows that $\rho$ is an unboundedness witness for $\Config{q}{m}$ also in $V_\Z$, and thus $\pismall\rho$ is an unboundedness witness for $s$. Now it is safe to remove from $\pismall$ any cycle $\theta$ with $\proj{I_\N}{\eff(\theta)} = \vec{0}$, which has no effect on the $\N$-counters. So we assume that no two configurations on $\pismall$ have the same values on counters in $I_\N$, and moreover, all configurations on $\pismall$ except the last one are $C$-small. By \cref{clm:small-configurations} we deduce that $|\pismall| \le n \cdot (d \cdot C)^k$. Hence, $|\pismall \rho| \le n \cdot (d \cdot C)^k + L_{k-1} \le L_{k}$.

\textbf{Case 2.} Assume all configurations on $\pi$ are $C$-small. Split $\pi$ into $\pi_1\pi_2$ where $\pi_2$ is a cycle with $\eff(\pi_2) > \vec{0}$. If $\proj{I_\N}{\eff(\pi_2)} > \vec{0}$, then by repeating $\pi_2$ sufficiently many times one eventually encounters a $C$-large configuration, which reduces to Case 1. Hence, we assume next that $\proj{I_\N}{\eff(\pi_2)} = \vec{0}$. The strategy is to encode the $\N$-counters into states and reduce to the base case where $\ressumdim(V_\Z) = 0$.

Let $S := \{(q, \proj{I_\N}{\vec{y}}) \mid q(\vec{y}) \text{ is a $C$-small configuration reachable from }s\}$. By \cref{clm:small-configurations} we have $|S| \le n \cdot (d \cdot C)^k$. We construct a VASS $V'$ whose states are $S$ and whose transitions are of the form $(p, \vec{m}) \trans{\proj{I_\Z}{\vec{a}}} (q, \vec{n})$ if there is a transition $p \trans{\vec{a}} q$ in $V_\Z$ with $\proj{I_\N}{\vec{a}} = \vec{n} - \vec{m}$. We set all counters in $V'$ to be integer counters, so $\ressumdim(V') = 0$. Now the unboundedness witness $\pi = \pi_1\pi_2$ can be mapped to a run $\pi' = \pi_1'\pi_2'$ in $V'$. Moreover, $\pi_2'$ is still a cycle in $V'$ as we have assumed that $\proj{I_\N}{\eff(\pi_2)} = \vec{0}$. Also, $\eff(\pi_2') = \proj{I_\Z}{\eff(\pi_2)} > \vec{0}$. Hence $\pi'$ is an unboundedness witness in $V'$. Using the induction base, there is an unboundedness witness $\rho$ with the same source as $\pi'$ such that $|\rho| \le (5d^2|S|^2M)^{2g+2} \le (5d^2M(n \cdot (d \cdot C)^k)^2)^{2g+2} \le L_{k}$. Finally it is easy to observe that $\rho'$ can be mapped back to a run in $V_\Z$ that is also an unboundedness witness for $s$.


\section{Integer Reachability in Geometric Dimension}
\label{sec:z-reachability}

In this section we prove that in geometric dimension, \(\Z\)-reachability behaves the same as in the case when we parameterise by the number of counters. We have the following theorem:

\begin{theorem}
\label{thm:z-vass-sum}
Let $V$ be a VASS of geometric dimension $g$. For any two $\Z$-configurations $s$ and $t$ such that $t$ is reachable from $s$ by a $\Z$-run, there exists also a $\Z$-run $\pi$ from $s$ to $t$ such that $|\pi| \leq \oo(M)^{6g+1}$ where $M = \max(\size(V), \infnorm{s}, \infnorm{t})$. 
\end{theorem}

\begin{proof}
Let $s = q(w)$ and $t = q(w')$. Let $\sigma$ be a $\Z$-run between $s$ and $t$. Similar to what we have done in the base case of \cref{lem:boundedness-vass-z} (see Page \pageref{sec:boundedness-base-case}), we express a $\Z$-run between $s$ and $t$ as a solution of a Diophantine system of linear equations, and rely on Lemma~\ref{lem:linear_equations}. Let $Q' \subseteq Q$ and $T' \subseteq T$ be the subsets of states and transitions that appear in $\sigma$. Simple cycles that use only transitions from $T'$ we call $T'$-cycles. The $\Z$-run $\sigma$ decomposes into a $\Z$-run
$\sigma_0$ that visits all states of $Q'$ and set $\mathcal{S}$ of simple $T'$-cycles. Choose the shortest such $\sigma_0$. Observe, that $\sigma_0$ visits each state at most $|Q'| \leq |Q| \leq M$ times, as otherwise it could be shortened, and therefore its effect has norm at most $M^3$. For each cycle $\theta \in \mathcal{S}$ we have $\infnorm{\eff(\theta)} \leq M^2$ as it is a simple cycle. Observe, that $t - s - \eff(\sigma_0) \in \intcone(\eff(\mathcal{S}))$. Therefore, because of \Cref{thm:caratheodory-integer} we can choose $\mathcal{C} \subseteq \mathcal{S}$ such that $|\mathcal{C}| \leq 2d \log(4dM) \leq 8d^2M \leq 8M^3$ such that $t-s-\eff(\sigma_0) \in \intcone(\eff(\mathcal{C}))$. 
We define a system $U$ of $d$ linear equations (one for each dimension), whose unknowns $x_\theta$ correspond to $T'$-cycles $\theta$ from $\mathcal{C}$:
$$\sum_{\theta \in \mathcal{C}} x_{\theta} \cdot \eff(\theta) = t - s - \eff(\sigma_0)$$
The system has a nonnegative integer solution, namely the one obtained from the fact that $t - s - \eff(\sigma_0) \in \intcone(\eff(\mathcal{C}))$. The system $U$ can be expressed as $A \cdot x = b$ for some matrix $A$. Observe, that rank of the matrix $A$ is at most $g$ as columns of $A$ are effects of simple cycles in $V$. Hence we can choose at most $g$ linearly independent equations from the system $U$ and create a system $U'$, which can be expressed as $A' \cdot x = b'$ for some $A' \in \Z^
{g \times |\mathcal{C}|}$ and $b' \in \Z^g$.  As all coefficients of $U'$ are bounded by $\oo(M^3)$, by Lemma~\ref{lem:linear_equations} the system has a solution of norm
$\oo((|\mathcal{C}|M^3))^g = \oo(M)^{6g}$. The solution yields a $\Z$-run $\pi$ from $s$ to $t$ of length $|\sigma_0| + M \cdot \oo(M)^{6g} = \oo(M)^{6g+1}$, consisting, of $\sigma_0$ with attached all cycles $\theta \in \mathcal{C}$ (this is possible, as $\sigma_0$ visits all states used by the cycles), each $\theta$ iterated $x_\theta$ times. This completes the proof.
\end{proof}

\section{Tower-Hardness of Reachability in SCC Dimension 4}
\label{sec:tower-max-dim-4-lower-bound}
In this section we prove the following theorem.

\begin{restatable}{theorem}{maxDimFourTowerHard}
    \label{thm:max-dim-4-tower-hard}
    Reachability in VASS of SCC dimension $\maxdim = 4$ is \tower-hard.
\end{restatable}

The proof is based on a modification of \cite[Theorem 4]{CzerwinskiO22} which showed \tower-hardness for 8-VASS. The general strategy is to reduce from the \tower-bounded reachability problem in two-counter machines. Two-counter machines can be viewed as 2-VASSes enhanced with the ability of testing equality with zero on each counter, while \tower-boundedness requires values of both counter to stay below $\tower(n)$, where $n \in \N$ is part of the input. We have to implement zero-tests on bounded counters in VASS without increasing the SCC dimension too much. We will rely on a useful gadget called \emph{multiplication triples} \cite{CzerwinskiLLLM19}.

A multiplication triple consists of three VASS counters $\Counter{b}, \Counter{c}, \Counter{d}$ holding values $B, C, B\cdot C$. For any counter $\Counter{x}$ bounded by $B$ we can simulate $C/2$ zero tests on it. The procedure is as follows: we try to move values from $\Counter{b}$ to $\Counter{x}$ and then back to $\Counter{b}$, in the meantime we decrease $\Counter{d}$ by the amount of value actually moved. Finally, $\Counter{c}$ is decreased by two. Since \(\Counter{x}\) is bounded by \(B\), \(\Counter{d}\) was decreased by at most \(2B\). To maintain \(\Counter{d} \leq \Counter{b} \cdot \Counter{c}\), we must have been able to decrease \(\Counter{d}\) by \(2B\), i.e.\ we must have \(\Counter{x}=0\). This can be extended to zero-testing multiple counters $\Counter{x}_1, \Counter{x}_2, \ldots, \Counter{x}_k$, by moving values between $\Counter{b}$ and these counters like a chain, see \cite{CzerwinskiO22}. 
We call $\Counter{d}$ the sensor counter of this multiplication triple. It is important to notice that we always maintain the invariant $\Counter{b} + \Counter{x}_1 + \Counter{x}_2 + \cdots + \Counter{x}_k = B$, wherefore counter \(\Counter{b}\) does not increase the SCC dimension \(\maxdim\), i.e.\ a multiplication triple increases \(\maxdim\) only by \(2\).

We shall construct a VASS of two parts: the first part generates a multiplication triple $\Counter{b}, \Counter{c}, \Counter{d}$ with $\Counter{b} = \tower(n)$, and the second part simulates the 2-counter machine using this triple. As noted above, the SCC dimension in the second half is at most $2 + 2 = 4$, however the first part requires more care.

In order to generate the tower-sized multiplication triple, a so-called \emph{amplifier} was utilized in \cite{CzerwinskiO22}. An amplifier consumes a multiplication triple with values $(\Counter{b}, \Counter{c}, \Counter{d}) = (B, C, B \cdot C)$ and generates a new one with values $(\Counter{b}', \Counter{c}', \Counter{d}') = (2^B, C', 2^B \cdot C')$. Intuitively, it starts with a triple $(\Counter{b}', \Counter{c}', \Counter{d}') = (1, C', C')$ and doubles $\Counter{b}'$ and $\Counter{c}'$ for $B$ times. Each doubling requires zero-tests, which are implemented using the triple $(\Counter{c}, \Counter{b}, \Counter{d})$, with the roles of $\Counter{b}, \Counter{c}$ swapped to ensure the doubling is performed exactly $B$ times. The amplifier needs zero-tests on three counters: $\Counter{b}', \Counter{d}'$ and an additional counter used to store the result of each doubling. This gives SCC dimension $3 + 2 = 5$, which is not optimal. To improve this, notice that if we have two identical copies of the input triple $(\Counter{b}, \Counter{c}, \Counter{d})$, we can amplify $\Counter{b}'$ and $\Counter{d}'$ separately. Now each amplification procedure yields SCC dimension $4$. Of course, we need to generate two identical copies of $\Counter{b}', \Counter{c}', \Counter{d}'$ to continue this amplification. But adding copies of counters does not change the SCC dimension. Notice that the amplifier succeeds if and only if the sensor counters in the two copies of its input triples become zero at the end.

By concatenating $n$ copies of the amplifier, we are able to generate the $\tower(n)$-multiplication triple. We remark that each amplifier should use fresh new counters for its output. This does not increase the SCC dimension as each amplifier works in its own SCC. This way, we are allowed to check if all amplifiers and multiplication triples function correctly by checking whether all the sensor counters equal zero at the end of the run, which can be encoded into the reachability problem. Therefore, we conclude that reachability in VASS of $\maxdim = 4$ is \tower-hard. A more detailed proof of \cref{thm:max-dim-4-tower-hard} can be found in \cref{app:tower-hardness-max-4}.

\section{Future Research}
\label{sec:future}
In~\cref{sec:coverability}, we proved that coverability in VASS parameterised by geometric dimension $\sumdim$ is witnessed by double exponential length runs.
Specifically, we argued that the exponent only depends on $\sumdim$ (\cref{thm:coverability-sum-dim}).
As we touched on in~\cref{rem:coverability-max-dim}, this does not allow us to obtain runs with less than double exponential length for instances of coverability in VASS with fixed SCC dimension $\maxdim$.
However, we conjecture that the coverability problem in VASS with fixed SCC dimension can be witnessed by exponential length runs.

\begin{conjecture}\label{conj:coverability-max-dim-pspace}
    For every fixed $\maxdim \in \N$,
    the coverability problem in VASS with SCC dimension $\maxdim$ is in \class{PSpace}.
\end{conjecture}

Unlike for geometric dimension, one cannot hope for sub-exponential length runs for coverability in VASS with fixed SCC dimension; see the following example.

\begin{example}\label{ex:exp-path-length}
    Consider a $d$-VASS with $d$ states $q_1, \ldots, q_d$.
    For every $i \in [d-1]$, we have a transition $q_i \trans{\vec{0}} q_{i+1}$.
    Furthermore, for every $i \in [d] \setminus \{1\}$, there is a self-loop
    $q_i \trans{-2\vec{e}_{i-1} + \vec{e}_i} q_i$; and we also add the self-loop $q_1 \trans{\vec{e}_1} q_1$. The shortest run from $q_1(\vec{0})$ that covers $q_d(\vec{e}_d)$ is
    \begin{multline*}
        q_1(\vec{0}) \tran q_1(2^{d-1} \vec{e}_1) \tran q_2(2^{d-1} \vec{e}_1) \tran q_2(2^{d-2} \vec{e}_2) \\
        \tran q_3(2^{d-2} \vec{e}_2) \tran q_3(2^{d-3} \vec{e}_3) \tran \ldots \tran q_{d-1}(2 \vec{e}_{d-1}) \tran q_d(2\vec{e}_{d-1}) \tran q_d(\vec{e}_d).
    \end{multline*}
    The length of this run is exponential in $d$. 
\end{example}

One attempt to prove~\cref{conj:coverability-max-dim-pspace} would be the following \emph{strong-linearity conjecture} in VASS parameterised by geometric dimension.
\cref{conj:strong-linearity} is a substantial strengthening of~\cref{thm:coverability-sum-dim} which has already been stated for VASS parameterised by the standard dimension $d$~\cite{Hack76,Jecker22,PilipczukSS25}.

\begin{conjecture}\label{conj:strong-linearity}
    There is some exponential function $f: \N \to \N$ such that, if there is a run in $V$ from $\Config{s}{x}$ that covers $\Config{t}{y}$, then there exists a run in $V$ from $\Config{s}{x}$ that covers $\Config{t}{y}$ of length at most $\infnorm{\vec{y}} \cdot \poly(d,n,M)^{f(\sumdim)}$.
\end{conjecture}

\begin{proof}[Proof of~\cref{conj:coverability-max-dim-pspace} assuming~\cref{conj:strong-linearity}]
    It is well-known that a bound $L$ on the length of a covering run implies that the backward coverability algorithm stabilises after at most $L$ many steps (and vice versa). 
    We consider the backward coverability algorithm, starting from the target configuration $\Config{t}{y}$. We first focus on the last SCC. 
    By~\cref{conj:strong-linearity}, one can observe that running the backwards coverability algorithm in the last SCC takes at most $\infnorm{\vec{y}} \cdot \poly(d,n,M)^{f(g)}$ steps.
    This means that minimal configurations that can cover the target configuration are of size at most $L \cdot M \cdot \poly(d,n,M)^{f(g)} = L \cdot \poly(d,n,M)^{f(g)}$.
    Then, we repeat this argument for each SCC in reverse order to obtain a bound on the size of the minimal configurations that can cover the target configuration (in particular the length of the shortest covering run) equal to $\infnorm{\vec{y}} \cdot \poly(d,n,M)^{|Q| \cdot f(\maxdim)}$, which for fixed $\maxdim$ is $\infnorm{\vec{y}} \cdot \poly(d, n, M)^{|Q|} = f_g(d, n, M, L)$, for some exponential function $f_g$.
    Hence coverability in VASS with fixed SCC dimension $\maxdim$ is witnessed by exponential length runs which allows us to conclude that coverability is in \class{PSpace}.
\end{proof}

Another direction for future research is the boundedness problem. Example \ref{ex:exp-path-length} can be adapted to boundedness rather easily, proving that similar to coverability the shortest run witnessing unboundedness is either exponential or doubly-exponential when the SCC dimension is fixed. We conjecture that it is actually exponential.

\begin{conjecture}\label{conj:exp-path-lenght-boundedness}
    For every $g \in \N$ there is an exponential function $f_g$ such that
    for every $d$-VASS $V = (Q, T)$ with $\maxdim(V) = g$,
    if there is an unboundedness witness for configuration $s$ then there is also one of length at most $f_g(d,n,M)$ where $n = |Q|$ and $M = \infnorm{T}$.
\end{conjecture}

Finally, while the general VASS reachability algorithm works for SCC dimension \(\maxdim\), in low SCC dimension no tight complexity bounds are known, as opposed to the parameterisations by the number of counters \(d\) \cite{BlondinFGHM15} and geometric dimension \(\sumdim\) \cite{Zheng25}. Another natural direction for future research is to study these complexities, we conjecture the following.

\begin{conjecture}\label{conj:max-dim-2-reach}
Reachability in VASS of SCC dimension $\maxdim = 2$ and $\maxdim = 3$ is in \pspace and \expspace, respectively.
\end{conjecture}

\bibliography{references}

\appendix 

\section{Proof of \cref{clm:linear_inequalities}}
\label{sec:app_preliminaries}

\clmLinearInequalities*

To prove it we need to recall the following result from~\cite{ilp}.

\begin{lemma}\label{lem:ilp}
    Let $A, B, C, D$ be $m \times n$, $m \times 1$, $p \times n$, and $p \times 1$ matrices respectively, all with integer entries.  
    Let the rank of $A$ be $r$, and let $s$ be the rank of the $(m+p) \times n$ matrix $\begin{pmatrix}
    A \\
    C
    \end{pmatrix}
    $. Let $M$ be an upper bound for the absolute values of all $(s-1)\times(s-1)$ or $s\times s$ subdeterminants of the $(m+p)\times(n+1)$ matrix
    $\begin{pmatrix}
    A & B \\
    C & D
    \end{pmatrix}$,
    which are formed with at least $r$ rows from $(A\; B)$. If the systems
    $A\vec{x} = B$ and $C\vec{x} \ge D$
    has a common integer solution, then it has one whose coefficients are bounded in absolute value by $(n+1)M$.
\end{lemma}

Now we are ready to prove \Cref{clm:linear_inequalities}.

\begin{proof}[Proof of~\cref{clm:linear_inequalities}]
    Let $\vec{x} \in \N^d$ be a solution of the system $A \cdot \vec{x} \ge \vec{b}$. Observe that $A \cdot \vec{x}$ belongs to the cone generated by columns of the matrix $A$. By \Cref{thm:caratheodory-rational} (Carath\'eodory's Theorem for Rational Cones) there exists $\vec{y} \in \Q_{\geq 0}^d$ with at most $r$ nonzero coordinates such that $A \cdot \vec{x} = A \cdot \vec{y}$. Since $\vec{b} \geq  \vec{0}$ we get that there exists $n \in \N$ such that $A \cdot (n \cdot \vec{y}) \geq \vec{b}$ and $n \cdot \vec{y} \in \N^d$.

    Consider the matrix $A'$ consists of columns corresponding to nonzero coordinates of $\vec{y}$. Then the system $A' \cdot \vec{x}' \geq \vec{b}$ has a nonnegative solution (for instance $n \cdot \vec{y}$ projected to its non-zero coordinates), and each such solution can be mapped to a solution of $A \cdot \vec{x} \geq \vec{b}$ by setting other variables to zero. 

    nonnegative solutions of $A' \cdot \vec{x}' \geq \vec{b}$ is a common solution of $N \cdot \vec{x}' = \vec{0}$, $C \cdot \vec{x}' \geq \vec{b}'$ where $N$ is the zero matrix and $C \cdot \vec{x}' \geq \vec{b}'$ is the system $A' \cdot \vec{x}' \geq \vec{b}$ extended with inequalities ensuring $\vec{x}' \ge \vec{0}$. Notice that $\rank(N) = 0$ and $\rank(C) \leq r$ as $C$ has at most $r$ columns. Recall that, by Hadamard's inequality, determinants of $k\times k$ matrices is bounded by $B^kk^{\frac{k}{2}}$ where $B$ is the bound in the absolute value of entries of the matrix. Hence, by \Cref{lem:ilp} we get that the system $N \cdot \vec{x}' = \vec{0}$, $C \cdot \vec{x}'  \geq b'$ has a nonnegative solution whose coefficients are bounded by $(r+1)N^{r}r^{\frac{r}{2}} \leq (r+1)(Nr)^r$.

    Therefore, $A' \cdot \vec{x}' \geq \vec{b}$ (and hence $A \cdot \vec{x} \geq \vec{b}$) has a nonnegative integer solution of max-norm at most $(r+1)(rN)^{r}$.
\end{proof}

\section{Missing Proofs of~\cref{sec:coverability}}
\label{app:coverability}
\coverabilitySumDim*

\begin{proof}\label{proof:coverability-sum-dim}
        By~\cref{lem:rackoff}, we know that if there is a run in $V$ from $\Config{s}{x}$ that covers $\Config{t}{y}$, then there exist a run in $V$ from $\Config{s}{x}$ that covers $\Config{t}{y}$ of length at most $L_{\sumdim}$.
        Thus, it only remains to prove that $L_{\sumdim}$ is indeed bounded by $\poly(d,n,\infnorm{y},M)^{f(\sumdim)}$.
        Let $A = 4ndM(\infnorm{\vec{y}} + 1)$ and recall that $f(g) = (g+1)^{g+1}$, wlog. we assume $M > 0$.
        We will prove that \hbox{$L_g \leq A^{f(g)}$} by induction on $g$.
        For $g = 0$, by definition $L_0 = n-1$, and we know that $(4ndM(\infnorm{\vec{y}} +1))^{f(0)} \geq 4n > n-1$.
        For $g \geq 1$, assume that $L_{g-1} \leq A^{f(g-1)}$.
        Now, by definition, we know that $L_g = n(d(M\cdot L_{g-1} + \infnorm{\vec{y}}))^g + L_{g-1}$.
        Hence,
        \begin{align*}
            L_g & \leq 2n(d(M\cdot L_{g-1} + \infnorm{\vec{y}}))^g \leq (2nd (M \cdot L_{g-1} + \infnorm{\vec{y}}))^g \\
            & \leq (4ndM \cdot \infnorm{\vec{y}} \cdot L_{g-1})^g \leq (A \cdot L_{g-1})^g \leq (A \cdot A^{f(g-1)})^g \\
            & = A^{(g^g+1)\cdot g} \leq A^{(g+1)^g \cdot g} \leq A^{(g+1)^{g+1}} = A^{f(g)},
        \end{align*}
        where the third inequality holds because $a + b \leq 2ab$ for all $a, b \in \N_{>0}$ and
        the fifth one holds by induction.
        We conclude as $L_{\sumdim} \leq 4ndM(\infnorm{\vec{y}} + 1)^{f(\sumdim)}$ is bounded by $\poly(d,n,M,\infnorm{\vec{y}})^{f(\sumdim)}$.
    \end{proof}

\pathDifference*

\begin{proof}
    First, consider the case when $p$ and $q$ are in the same strongly connected component.
    Let $\tau$ be any path from $q$ back to $p$.
    Then $\rho_1 \tau$ is a cycle from $p$ to $p$, and so is $\rho_2 \tau$. 
    Thus the difference of the effects of $\rho_1 \tau$ and $\rho_2 \tau$ belongs to $\cyclespace{V}$.
    Clearly, the difference between effects of $\rho_1 \tau$ and $\rho_2 \tau$ is the same as the difference between the effects of $\rho_1$ and $\rho_2$.
    
    We now proceed to the general case (when $p$ and $q$ are not necessarily in the same SCC).
    Recall that $V$ consists of a series of strongly connected components $V_1, \ldots, V_m$ that are connected by single transitions.
    Without loss of generality, suppose that $p$ is in $V_1$ and $q$ is in $V_m$.
    We know that both $\rho_1$ and $\rho_2$ take the same transitions between the strongly connected components.
    Let the effect of $\rho_i$ in the component $V_j$ be $\vec{r}_{i,j}$.
    The difference between the effects of $\rho_1$ and $\rho_2$ is
    \[
        \sum_{j=1}^m \vec{r}_{1,j} - \sum_{j=1}^m \vec{r}_{2,j} = \sum_{j=1}^m (\vec{r}_{1,j} - \vec{r}_{2,j}).
    \]
    However, for every $j \in [m]$, it is true that $\vec{r}_{1,j} - \vec{r}_{2,j}$ is the difference of the effects of two paths between the same pair of states inside a single SCC. 
    By our earlier argument, we know that $\vec{r}_{1,j} - \vec{r}_{2,j} \in \cyclespace{V}$.
    Thus, since $\cyclespace{V}$ is a vector space, we know that $\sum_{j=1}^m (\vec{r}_{1,j} - \vec{r}_{2,j}) \in \cyclespace{V}$.
    This completes the proof.
\end{proof}

\smallConfigurations*
\begin{proof}
    For every state $q \in Q$, let us fix $\vec{r}_q$ to be the effect of an arbitrarily selected path from $p$ to $q$.

    Observe that there is at most one clean basis $B$ for any given choice of distinguished counter $K = \set{1, \ldots, k_{\sumdim}}$.
    Accordingly, there are at most $d \cdot (d-1) \cdot \ldots \cdot (d-\sumdim+1) \leq d^{\sumdim}$ many clean bases.
    
    Now, let us fix $q \in Q$ and let us fix a clean basis $B$ of $\cyclespace{V}$ whose distinguished coordinates are $K$ (notice that $\abs{K} = \sumdim$).
    Consider two configurations $\Config{q}{v}$, $\Config{q}{v'}$ that are reachable from $\Config{p}{u}$.
    By~\cref{clm:path-difference}, we can express $\vec{v} = \vec{u} + \vec{r}_q + \vec{w}$ and $\vec{v'} = \vec{u} + \vec{r}_q + \vec{w'}$ for $\vec{w}, \vec{w}' \in \cyclespace{V}$.
    
    Now, we wish to argue that $\vec{v} \neq \vec{v}'$ implies $\proj{K}{\vec{v}} \neq \proj{K}{\vec{v'}}$. Clearly, as $\vec{v} \neq \vec{v}'$, we have $\vec{w} \neq \vec{w}'$. Notice that since $\vec{w},\vec{w}' \in \cyclespace{V}$ it holds that $\vec{w} = \sum_{i=1}^\sumdim (\proj{K}{\vec{w}})[i] \cdot \vec{b}_i$ and $\vec{w'} = \sum_{i=1}^\sumdim (\proj{K}{\vec{w'}})[i] \cdot \vec{b}_i$. Hence, $\proj{K}{\vec{w}} \neq \proj{K}{\vec{w'}}$. Observe that $\proj{K}{\vec{v}} = \proj{K}{\vec{u}} + \proj{K}{\vec{r}_q} + \proj{K}{\vec{w}}$ and $\proj{K}{\vec{v'}} = \proj{K}{\vec{u}} + \proj{K}{\vec{r}_q} + \proj{K}{\vec{w'}}$. Therefore, as $\proj{K}{\vec{w}} \neq \proj{K}{\vec{w'}}$, we get $\proj{K}{\vec{v}} \neq \proj{K}{\vec{v'}}$.
    
    By the previous argument and the fact there are at most $d^{\sumdim}$ many clean bases, we conclude that the number of configurations $\Config{q}{v}$ that are reachable from $\Config{p}{u}$ is at most 
    \begin{equation*}
        \abs{\set{\proj{K}{\vec{v}} : \proj{K}{\vec{v}} \in [0, C-1]^{\sumdim}}} \cdot d^{\sumdim}
        = (C \cdot d)^\sumdim.
    \end{equation*}
    Thus, as there are $n = \abs{Q}$ many states, the total number of $C$-small configurations that are reachable from
     $\Config{p}{u}$ is at most $n\cdot (C \cdot d)^\sumdim$.
\end{proof}


\section{Missing Proofs of~\cref{sec:simultaneous-unboundedness}}
\label{app:simultaneous}
\noPump*
\begin{proof}
By~\Cref{lem:no-pump-concrete}, we know that if there exists a run
    from $\Config{s}{x}$ to $\Config{q}{u}$ in which, for every $i \in [d]$, there is a configuration $p_i(\vec{h}_i)$ with $\vec{h}_i[i] \geq H_\sumdim$, then there exists a run from $\Config{s}{x}$ to $\Config{q}{v}$ of length at most $L_\sumdim$, such that $\vec{v} \geq (G, \ldots, G)$. Thus, it only remains to prove that $L_\sumdim$ and $H_\sumdim$ are indeed bounded by $\poly(d,n,G,M)^{f(\sumdim)}$.
    
    Let $A = 2n \cdot (d+1) \cdot M (G+1)$ and recall that $f(g) = (g+1)^{g+1}$. We will prove that $L_\sumdim, H_\sumdim \leq A^{f(\sumdim)}$ by induction on $\sumdim$. 
    
    For $\sumdim = 0$, by definition $L_0 = n \cdot (d+1) \cdot M + G \leq A^{f(0)}$. 
    Similarly, $H_0 = n\cdot (d+1) \leq A^{f(0)}$. 
    
    For $\sumdim \geq 1$, assume that $L_{\sumdim-1}, H_{\sumdim-1} \leq A^{f(\sumdim-1)}$. Now, by definition we know that $L_\sumdim = n \cdot (d \cdot C_\sumdim)^\sumdim + L_{\sumdim-1}$ and $H_\sumdim = n \cdot M \cdot (d \cdot C_\sumdim)^{\sumdim} + H_{\sumdim - 1}$ where $C_\sumdim = M \cdot L_{\sumdim-1} + G$. Hence,
    \begin{multline*}
        L_\sumdim 
        = n\cdot (d \cdot C_\sumdim)^\sumdim + L_{\sumdim-1}
        = 2n \cdot (d \cdot (M \cdot L_{\sumdim-1} + G))^g \\
        \leq (A \cdot L_{\sumdim-1})^\sumdim  
        \leq (A \cdot A^{f(\sumdim-1)})^\sumdim 
        \leq A^{\sumdim + \sumdim^\sumdim \cdot \sumdim}
        \leq A^{\sumdim + \sumdim^{\sumdim+1}} 
        \leq A^{f(\sumdim)} 
    \end{multline*}
    and 
    \begin{align*}
        H_\sumdim 
        & = n \cdot M \cdot (d \cdot C_\sumdim)^\sumdim + H_{\sumdim-1} 
        = n \cdot M \cdot (d \cdot (M \cdot L_{\sumdim-1} + G))^g + H_{\sumdim-1} \\
        & \leq (A\cdot L_{\sumdim-1})^\sumdim + H_{\sumdim-1}
        \leq (A\cdot A^{f(\sumdim-1)})^\sumdim + A^{f(\sumdim-1)} \\
        & \leq (A \cdot A^{\sumdim^\sumdim})^g + A^{\sumdim^\sumdim}
        \leq 2A^{\sumdim + \sumdim^\sumdim\cdot \sumdim}
        \leq A^{1 + \sumdim + \sumdim^{\sumdim+1}}
        \leq A^{f(\sumdim)}.
    \end{align*}
    We conclude as $L_g, H_g \leq A^{f(\sumdim)}$ is bounded by $\poly(d,n,M,G)^{f(\sumdim)}$.
\end{proof}

\section{Missing Proofs of~\cref{sec:boundedness}}
\label{app:boundedness}
\boundedness*

\begin{proof}
    Notice that for a standard VASS $V$ we have $\ressumdim(V) = g$. Thus, we only need to verify that the bound $L_{g}$ given by \cref{lem:boundedness-vass-z} is of the form $\poly(d, n, M)^{f(\sumdim)}$. Recall that $D = (5d^2n^2M)^2$ is polynomial in $d, n, M$. Let $C = 2DdM$, it is also polynomial in $d, n, M$.
    
    We show by induction on $k$ the following claim.
    \begin{claim}\label{cl:bound-lkg}
        $L_{k} \leq C^{(4k+1)^k (g+2)^{k+1} - 1}$.
    \end{claim}
    
    \begin{claimproof}
        The base case where $k = 0$ holds
        since $L_{0} = D^{g+1} \leq C^{g+1}$.
        For the induction step we have
        \begin{align*}
            L_{k} & = (D(d \cdot M \cdot L_{k-1})^{4k})^{g+1} + L_{k-1}
            \leq (2DdM \cdot L_{k-1})^{4k(g+1)} = (C \cdot L_{k-1})^{4k(g+1)} \\
            & \leq (C \cdot C^{(4(k-1)+1)^{k-1} (g+2)^k - 1})^{4k(g+1)} = C^{((4(k-1)+1)^{k-1} (g+2)^k) \cdot 4k(g+1)} \\
            & \leq C^{(4k+1)^{k+1} (g+2)^{k+1} -1},
        \end{align*}
        as needed.
    \end{claimproof}
    
    Having Claim~\ref{cl:bound-lkg} we conclude that
    $L_{g} \leq C^{(4g+1)^g (g+2)^{g+1} - 1} \leq C^{(4g+2)^{2g+1}} = C^{f(g)}$, as required.
\end{proof}

\NCountersIngored*

\begin{claimproof}
    We prove by induction on the length of $\rho$. Suppose $\rho$ is a path from state $p$ to $q$. We have assumed that $\pi$ contains every transition of $V$, thus $\pi$ also visits the state $q$. Let $q(\vec{y}_\pi)$ be any configuration on the run $\pi$ whose state is $q$, and let $\pi'$ be the prefix of $\pi$ from $p(\vec{x})$ to $q(\vec{y}_\pi)$. 
    According to \cref{clm:path-difference}, $\vec{\delta} := \eff(\pi') - \eff(\rho)$ belongs to $\cyclespace{V}$. Hence, $\proj{I_\N}{\vec{\delta}} = \vec{0}$. Suppose $s = p(\vec{x})$ and let $\vec{y} := \vec{x} + \eff(\rho)$. Then $\proj{I_\N}{\vec{y}} = \proj{I_\N}{\vec{y}_\pi} \ge \vec{0}$. Hence, the ``target'' of $\rho$ starting from $s$ is a legal configuration of $V_\Z$. We conclude by induction that $s \xrightarrow{\rho} q(\vec{y})$ holds.
\end{claimproof}

\section{Tower-Hardness for Max-Geometric Dimension 4}
\label{app:tower-hardness-max-4}
In this appendix we give a more detailed proof of \cref{thm:max-dim-4-tower-hard}. As mentioned there, the proof is based on a modification of \cite[Theorem 4]{CzerwinskiO22}. Hence we assume the reader is familiar with the concepts in \cite{CzerwinskiO22}, especially the notion of \emph{counter programs} that can be turned into a VASS, and the technique of \emph{multiplication triples} (originates from \cite{CzerwinskiLLLM19}) that is used to implement restricted zero-tests. We start with the source of \tower-hardness.

\paragraph*{A \tower-hard Problem.}

In \cite{CzerwinskiO22} the \emph{Tower-bounded reachability problem} for 3-counter automata was used as the source of hardness. The problem can be formulated as follows:

\begin{quote}
    \textbf{Input}: A 3-counter automaton $\mathcal{A}$ and number $n \in \mathbb{N}$ given in unary\footnote{For Tower-hardness it doesn't matter if $n$ is given in unary or binary.}.
    
    \textbf{Question}: Does $\mathcal{A}$ have a $\Tower(n)$-bounded accepting\footnote{Acceptance by final states.} run.
\end{quote}

Note that an $f(n)$-bounded run can be defined in a slightly non-standard way: we require that at any configuration of the run, the \emph{sum} of values in 3 counters is bounded by $f(n)$. This definition favors the technique of multiplication triples.

We have no upper bound on the geometric dimension of $\mathcal{A}$ better than 3. And when simulating zero-tests of $\mathcal{A}$ with a multiplication triple, the geometric dimension of this part could be 5 (see discussions in the next part). So we need to use the 2-counter version of the above problem. 

\begin{lemma}
    Given a 3-counter automaton $\mathcal{A}$ and number $n \in \mathbb{N}$ one can construct in linear time a 2-counter automaton $\mathcal{A}'$ such that $\mathcal{A}$ has a $\Tower(n)$-bounded accepting run if and only if $\mathcal{A}'$ has a $7^{\Tower(n)}$ bounded accepting run.
\end{lemma}

\begin{proof}
    We first construct a 4-counter automaton $A''$ with counters $a, b, c, d$, such that it first computes $\Tower(n)$ and put it into $d$, then simulates $\mathcal{A}$ using the first 3 counters while maintaining the invariant $a + b + c = d$. Now $\mathcal{A}$ has a $\Tower(n)$-bounded accepting run $\iff$ $\mathcal{A}''$ has an accepting run. Moreover every run of $\mathcal{A}''$ is $\Tower(n)$-bounded. Now we transform $\mathcal{A}''$ into a 2-counter automaton $\mathcal{A}'$ using the encoding $(a, b, c, d) \mapsto 2^a3^b5^c7^d$. Clearly $\mathcal{A}''$ has an accepting run $\iff$ $\mathcal{A}'$ has a $7^{\Tower(n)}$ bounded accepting run.
\end{proof}

\begin{theorem}
    \label{thm:tower-comp-of-tower-bounded-2-cm-reach}
    Tower-bounded reachability in 2-counter automata is Tower-complete.
\end{theorem}

\paragraph*{Multiplication Triples}

A multiplication triple is a triple $\mathcal{M}$ of 3 counters $(\Counter{b}, \Counter{c}, \Counter{d})$ holding values initially $(B, C, B \cdot C)$ for some $B, C \in \N$. It is shown in \cite{CzerwinskiO22} that $\mathcal{M}$ can be used to simulate $C/2$ zero-tests on counters with sum of values bounded by $B$. Suppose there are 2 counters $\Counter{x}, \Counter{y}$ to be zero-tested in a counter program $\mathcal{A}$ (the method generalizes to more counters naturally). We first modify $\mathcal{A}$ such that the invariant $\Counter{x} + \Counter{y} + \Counter{b} = B$ is maintained. For example, an instruction $\Counter{x} {+}{=} 3; \Counter{y} {-}{=} 1$ is replaced with $\Counter{x} {+}{=} 3; \Counter{y} {-}{=} 1; \Counter{b} {-}{=} 2$. This guarantees that the sum of $\Counter{x}$ and $\Counter{y}$ is at most $B$. Now we implement zero tests on $\Counter{x}$ and $\Counter{y}$ by the programs in \cref{fig:zero-test}. Notice that the invariant $\Counter{x} + \Counter{y} + \Counter{b} = B$ never breaks.

\begin{figure}[h]
    \centering
    \begin{minipage}[t]{0.4\textwidth}
        \begin{algorithm}[H]
            \SetKw{Loop}{loop:}
            \Loop \Counter{y}$--$; \Counter{x}$++$; \Counter{d}$--$
            
            \Loop \Counter{b}$--$; \Counter{y}$++$; \Counter{d}$--$
            
            \Loop \Counter{y}$--$; \Counter{b}$++$; \Counter{d}$--$
            
            \Loop \Counter{x}$--$; \Counter{y}$++$; \Counter{d}$--$
            
            \Counter{c} ${-}{=}~ 2$
        \end{algorithm}
        \centering 
    \end{minipage}
    \begin{minipage}[t]{0.4\textwidth}
        \begin{algorithm}[H]
            \SetKw{Loop}{loop:}
            \Loop \Counter{x}$--$; \Counter{y}$++$; \Counter{d}$--$
            
            \Loop \Counter{b}$--$; \Counter{x}$++$; \Counter{d}$--$
            
            \Loop \Counter{x}$--$; \Counter{b}$++$; \Counter{d}$--$
            
            \Loop \Counter{y}$--$; \Counter{x}$++$; \Counter{d}$--$
            
            \Counter{c} ${-}{=}~ 2$
        \end{algorithm}
        \centering 
    \end{minipage}
    \caption{Left: ZeroTest(\Counter{x}); Right: ZeroTest(\Counter{y})}
    
    \label{fig:zero-test}
\end{figure}

The counter $\Counter{d}$ can reach the value $0$ if and only if (i) exactly $C/2$ zero-tests have been performed in $\mathcal{A}$, and (ii) each zero-test passed successfully. Thus we call $\Counter{d}$ the \emph{sensor counter} of the triple $\mathcal{M}$.

Now we discuss how the geometric dimension of the counter program $\mathcal{A}$ changes after plugging in the multiplication triple $\mathcal{M}$. 

\begin{itemize}
    \item For the counter $\Counter{b}$, its update in every transition is a linear combination of updates of $\Counter{x}$ and $\Counter{y}$. So $\Counter{b}$ does not contribute to the geometric dimension.
    \item For the counter $\Counter{d}$, we look at the program ZeroTest(\Counter{x}) for example. The sum of effects of the 4 loops there induces a vector in the cycle space that has only one entry with value $-4$ corresponding to $\Counter{d}$. So $\Counter{d}$ contributes $1$ to the geometric dimension.
    \item The counter $\Counter{c}$ is not updated in any loop of the zero-test programs. But zero-test can occur in loops of $\mathcal{A}$. Hence $\Counter{c}$ also contributes $1$ to the geometric dimension in general.
\end{itemize}

Suppose $\mathcal{A}$ operates on $d$ counters. Let $\mathcal{A}_{\mathcal{M}}$ be the counter program obtained from $\mathcal{A}$ by implementing zero-tests using the multiplication triple $\mathcal{M}$. The above discussion implies that $\sumdim(\mathcal{A}_{\mathcal{M}}) \le d + 2$.

Therefore, once we have obtained a multiplication triple $\mathcal{M}$ with values $(B, C, BC)$ for $B = \Tower(n)$ and $C$ being sufficiently large, the simulation of any 2-counter automaton $\mathcal{A}$ yields a VASS $\mathcal{A}_{\mathcal{M}}$ of geometric dimension at most 4. In the following we focus on implementing a counter program $\mathcal{G}$ that generates a multiplication triple with $B = \Tower(n)$.

\paragraph*{The (Old) Amplifier}

Amplifier is a VASS (or counter program) that consumes a multiplication triple and produce a larger multiplication triple. For tower-hardness we will need an exponential amplifier $\textbf{ampl}$. It operates on 7 counters $\Counter{b}, \Counter{c}, \Counter{d}, \Counter{b'}, \Counter{c'}, \Counter{d'}, \Counter{t}$, such that if at the beginning we have 
\begin{equation}
    (\Counter{b}, \Counter{c}, \Counter{d}) = (B, C, BC) \text{ and all other counters}= 0
\end{equation}
then all complete runs in $\textbf{ampl}$ ends with 
\begin{equation}
    (\Counter{b'}, \Counter{c'}, \Counter{d'}) = (2^B, C', 2^BC') \text{ {\bf if} all other counters }= 0 \text{ (which is guaranteed if } \Counter{d} = 0 \text{)}; 
\end{equation}
Moreover, for each $C' \in \mathbb{N}$ there is such a complete run starting with some $C$.

The amplifier was implemented in \cite{CzerwinskiO22} as the program shown in \autoref{fig:old-amplifier}. We remark that we are using the input triple with the roles of $\Counter{b}$ and $\Counter{c}$ swapped. Clearly the \textbf{loop} from lines 3 through 5 contributes most to the geometric dimension. Notice that 6 counters are involved in this loop, among which 3 are counters from the multiplication triple. According to the discussion in the previous section, the effect of every transition maintains the invariant $\Delta(\Counter{b'}) + \Delta(\Counter{d'}) + \Delta(\Counter{t}) + \Delta(\Counter{c})) = 0$, thus the geometric dimension of this loop is 5. 

\begin{figure}[h]
    \begin{minipage}[t]{0.5\textwidth}
        \begin{algorithm}[H]
            \SetKw{Loop}{loop:}
            \SetKwBlock{LoopB}{loop}{}
            \Counter{b'} ${+}{=}~1$
            
            \Loop \Counter{c'} ${+}{=}~1$; \Counter{d'} ${+}{=}~1$
            
            \LoopB{
                \textbf{mult}(\Counter{b'}, \Counter{t}, $2^8$)
                
                \textbf{mult}(\Counter{d'}, \Counter{t}, $2^8$)
            }
            
            \Loop \Counter{c}${-}{-}$
        \end{algorithm}
        \centering 
    \end{minipage}
    \begin{minipage}[t]{0.5\textwidth}
        \begin{algorithm}[H]
            \SetKw{Loop}{loop:}
            \SetKwBlock{LoopB}{loop}{}            
            \Loop \Counter{x} ${-}{=}~1$; \Counter{t} ${+}{=}~1$
            
            \textbf{ZeroTest}$_{\Counter{c}, \Counter{b}, \Counter{d}}$(\Counter{x})
            
            \Loop \Counter{x} ${+}{=}~c$; \Counter{t} ${-}{=}~1$
            
            \textbf{ZeroTest}$_{\Counter{c}, \Counter{b}, \Counter{d}}$(\Counter{t})
        \end{algorithm}
        \centering 
    \end{minipage}
    \caption{Left: Implementation of \textbf{ampl}; Right: Implementation of the procedure \textbf{mult}(\Counter{x}, \Counter{t}, $c$)}
    \label{fig:old-amplifier}
\end{figure}

The program $\mathcal{G}$ generating a multiplication triple with $B = \Tower(n)$ repeats the above amplifier \textbf{ampl} for $n$ times. Each amplifier use the previous output as its input, and use fresh counters as its output (we do not care much about the number of counters here). We are able to check if the results of all amplifiers are correct, by checking if all the sensor counters of the involved multiplication triples become zero at the end. Notice that SCCs of these amplifiers are disjoint. We conclude that $\maxdim(\mathcal{G}) \le 5$. This is not optimal as it implies \tower-hardness only for SCC dimension $5$.

\paragraph*{The Modified Amplifier}

The implementation of exponential amplifier in \autoref{fig:old-amplifier} is not optimal for geometric dimension as the counters \Counter{b'} and \Counter{d'} are multiplied in the same loop. It turns out that the two multiplications can be separated into two loops while ensuring they are executed the same number of iterations. The key is two use two identical copies of multiplication triples.

We will implement an amplifier $\textbf{ampl}'$ that operates on counters $\Counter{b}_1, \Counter{c}_1, \Counter{d}_1$, $\Counter{b}_2, \Counter{c}_2, \Counter{d}_2$, $\Counter{b}', \Counter{c}', \Counter{d}', \Counter{t}$. Assuming we have, at the beginning,
\begin{equation}
    (\Counter{b}_1, \Counter{c}_1, \Counter{d}_1) = (\Counter{b}_2, \Counter{c}_2, \Counter{d}_2) = (B, C, BC) \text{ and all other counters }= 0
\end{equation}
then all complete runs in $\textbf{ampl}$ ends with 
\begin{equation}
    (\Counter{b'}, \Counter{c'}, \Counter{d'}) = (2^B, C', 2^BC') \text{ {\bf if} all other counters}= 0 ~ (\text{guaranteed if } \Counter{d}_1 = \Counter{d}_2 = 0)
\end{equation}
Moreover for each $C' \in \mathbb{N}$ there is such a complete run starting with some $C$. We remark that in order to make $\textbf{ampl}'$ repeatable, we shall view each of the output counters $\Counter{b}', \Counter{c}', \Counter{d}'$ as identical copies of two counters. This does not increase the geometric dimension.

See~\autoref{fig:mod-amplifier} for the implementation of $\textbf{ampl}'$. 
Recall that each triple $(B, C, BC)$ allows $B/2$ zero tests. Thus if at the end we have $\Counter{d}_1 = \Counter{d}_2 = 0$ then both loop of multiplications are executed precisely $B / 4$ iterations, this implies that $\Counter{b}'$ and $\Counter{d}'$ are correctly multiplied by $2^B$.

\begin{figure}[h]
    \begin{minipage}[t]{\textwidth}
        \begin{algorithm}[H]
            \SetKw{Loop}{loop:}
            \SetKwBlock{LoopB}{loop}{}
            \Counter{b'} ${+}{=}~1$
            
            \Loop \Counter{c'} ${+}{=}~1$; \Counter{d'} ${+}{=}~1$
            
            \Loop
                \textbf{mult}(\Counter{b'}, \Counter{t}, $2^4$) with zero tests using $\Counter{c}_1, \Counter{b}_1, \Counter{d}_1$
                
            \Loop $\Counter{c}_1{-}{-}$
            
            \Loop
                \textbf{mult}(\Counter{d'}, \Counter{t}, $2^4$) with zero tests using $\Counter{c}_2, \Counter{b}_2, \Counter{d}_2$
            
            \Loop $\Counter{c}_2{-}{-}$
        \end{algorithm}
        \centering 
    \end{minipage}
    \caption{Implementation of $\textbf{ampl}'$}
    \label{fig:mod-amplifier}
\end{figure}

Each loop of multiplication now involves only 5 counters, among which 3 are from a multiplication triple. According to previous discussion the geometric dimension of each loop is 4. By plugging in the modified amplifier we conclude:

\maxDimFourTowerHard*

\end{document}